\documentclass[12pt]{article}

\title{Optimal Sequential Contests\thanks{
I am grateful to
Debraj Ray, four anonymous referees,
Stefano Barbieri,
Jeff Ely,
Andrea Gallice,
Dino Gerardi,
Marit Hinnosaar,
Johannes H\"orner,
Martin Jensen,
Kai Konrad,
Dan Kovenock,
Ignacio Monz\'{o}n,
Marco Ottaviani,
Alessandro Pavan,
Alex Possajennikov,
Marco Serena,
Ron Siegel,
Andy Skrzypacz,
and
Jean Tirole
as well as audience members at seminars at conferences
for their suggestions.
}}
\usepackage{datetime}
\newdateformat{monyear}{\monthname[\THEMONTH] \THEYEAR}
\date{
  \monyear\today
  \ifx\shortversion\undefined
    \thanks{The latest version of the paper is available at
    \href{https://toomas.hinnosaar.net/contests.pdf}{\url{toomas.hinnosaar.net/contests.pdf}}. A previous version of this paper was circulated under the title ``Dynamic common-value contests''.}\\
      {\small First version: December 2016}
  \else
  \fi
}

\author{
Toomas Hinnosaar%
\thanks{
University of Nottingham and CEPR, \href{mailto:toomas@hinnosaar.net}{\url{toomas@hinnosaar.net}}. %
}%
}
\usepackage[T1]{fontenc}
\usepackage{lmodern}
\usepackage{hyperref}
\usepackage{pgf,tikz}
\usepackage{natbib}
\usepackage{enumerate}
\usepackage{amssymb,amsmath,amsthm}
\usepackage{thmtools,thm-restate}
\usepackage{multirow}
\declaretheorem[name=Proposition]{proposition}
\declaretheorem[name=Theorem]{theorem}
\declaretheorem[name=Corollary]{corollary}
\declaretheorem[name=Lemma]{lemma}

\declaretheorem[name=Assumption]{assumption}

\newcommand{\dd}[3][]{\frac{d^{#1} #2}{d {#3}^{#1}}}

\newcommand{\R}{\mathbb R}

\newcommand{\uX}{\underline{X}}

\newcommand{\N}{\mathbb N}
\newcommand{\bn}{\mathbf{n}}
\newcommand{\bx}{\mathbf{x}}
\newcommand{\hn}{\widehat{n}}
\newcommand{\hbn}{\mathbf{\hn}}

\newcommand{\hf}{\widehat{f}}
\newcommand{\hX}{\widehat{X}}

\newcommand{\oX}{\overline{X}}

\usepackage{nicefrac}

\usepackage{cleveref}

\usepackage{anysize}
\usepackage{setspace}
\usepackage[titletoc,title]{appendix}
\usepackage{caption}
\usepackage{subcaption}

\newcommand{\gX}{\mathcal{X}}

\usepackage{color}

\usepackage{thmtools}
\makeatletter
\newcommand{\myref}[1]{\cref{#1}\mynameref{#1}{\csname r@#1\endcsname}}
\newcommand{\Myref}[1]{\Cref{#1}\mynameref{#1}{\csname r@#1\endcsname}}
\def\mynameref#1#2{%
    \begingroup
    \edef\@mytxt{#2}%
    \edef\@mytst{\expandafter\@thirdoffive\@mytxt}%
    \ifx\@mytst\empty\else
    \space(\nameref{#1})\fi
    \endgroup
}
\makeatother

\newcommand{\prt}{\mathcal{I}}

\crefname{condition}{condition}{conditions}
\Crefname{condition}{Condition}{Conditions}
\crefname{assumption}{assumption}{assumptions}
\Crefname{assumption}{Assumption}{Assumptions}
\crefname{figure}{figure}{figures}
\crefname{equation}{equation}{equations}

\newcommand{\players}{\mathcal{N}}

\DeclareMathOperator{\sgn}{sgn}
\newcommand{\prts}{\boldsymbol{\mathcal{I}}}

\newcommand{\bS}{\mathbf{S}}

\usepackage{multibib}

\begin{document}
\maketitle

\begin{abstract}
I study sequential contests where the efforts of earlier players may be disclosed to later players by nature or by design. The model has a range of applications, including rent seeking, R\&D, oligopoly, public goods provision, and tragedy of the commons. I show that information about other players' efforts increases the total effort. Thus, the total effort is maximized with full transparency and minimized with no transparency. I also show that in addition to the first-mover advantage, there is an earlier-mover advantage. Finally, I derive the limits for large contests.
\end{abstract}

\ifx\shortversion\undefined
\emph{JEL}: C72, C73, D72, D82, D74

\emph{Keywords}: contest design, oligopoly, public goods, rent-seeking, R\&D
\else
\fi

\section{Introduction} \label{S:intro}

Many economic interactions have contest-like structures, with payoffs that increase in players' own efforts and decrease in the total effort. Examples include oligopolies, public goods provision, tragedy of the commons, rent seeking, R\&D, advertising, and sports. Most of the literature assumes that effort choices are simultaneous. Simultaneous contests have convenient properties: the equilibrium is unique, is in pure strategies, and is relatively easy to characterize.

In this paper, I study contests where the effort choices are not necessarily simultaneous. In most real-life situations some players can observe their competitors' efforts and respond appropriately to those choices. However, earlier movers can also anticipate these ensuing responses and therefore affect the behavior of later movers. Each additional period in a sequential contest adds complexity to the analysis, which might explain why most previous studies focus on simultaneous or two-period models. In this paper, I characterize equilibria for an arbitrary sequential contest and analyze how the information about other players' efforts influences the equilibrium behavior.

Contests may be sequential by nature or by design. For example, in rent-seeking contests, firms lobby the government to achieve market power. One tool that regulators can use to minimize such rent-seeking is a disclosure policy. A non-transparent disclosure policy would lead to simultaneous effort choices, but a full transparency policy would lead to a fully sequential contest. There may be potentially intermediate solutions as well, where such information is revealed only occasionally. Over the last few decades, many countries have introduced new legislation regulating transparency in lobbying activity. This list includes the United States (Lobbying Disclosure Act, 1995; Honest Leadership and Open Government Act, 2007), the European Union (European Transparency Initiative, 2005), and Canada (Lobbying Act, 2008). However, there are significant cross-country differences in regulations. For example, lobbying efforts in the U.S. must be reported quarterly, whereas in the E.U., reporting occurs annually and on a more voluntary basis.

Another classic example of a contest is research and development (R\&D), where the probability of a scientific breakthrough is proportional to agents' research efforts. The question is how to best organize the disclosure rules to maximize aggregate research efforts. In some academic fields it is common to present early findings in working papers and conferences. In other fields these efforts are kept confidential until the work has been vetted and published in a journal. Similarly, when announcing an R\&D contest, the organizer can choose a transparency level: whether to use a public leaderboard or perhaps keep the entries secret until the deadline.

I analyze a model of sequential contests and provide two main results. First, I characterize all equilibria for any given sequential contest, i.e., for any fixed disclosure rule. The standard backward-induction approach requires finding best-response functions every period and substituting them recursively. This solution method is not generally tractable or even feasible. Instead, I introduce an alternative approach, in which I characterize each best-response function by its inverse. This method pools all the optimality conditions into one necessary condition and solves the resulting equation just once. I prove that for any contest the equilibrium exists and is unique. Importantly, the characterization theorem shows how to compute the equilibrium.

In the second main result I show that the information about other players' efforts strictly increases the total effort. Consequently, the optimal contest is always one of the extremes. When efforts are desirable (as in R\&D competitions), the optimal contest is one with full transparency. When the efforts are undesirable (as in rent-seeking), the optimal contest is one with hidden efforts. The intuition behind this result is simple. Near the equilibrium, the players' efforts are strategic substitutes. Therefore, players whose actions are disclosed have an additional incentive to exert effort to discourage later players' efforts. If the discouragement effect were strong enough to reduce the total effort, this would offer profitable deviations for some players. Therefore, near the equilibrium the discouragement effect is less than one-to-one. It increases earlier-movers efforts more than it reduces later-movers efforts, therefore increasing total effort. While there could be indirect effects that change the conclusions, I show that near the equilibrium efforts are higher-order strategic substitutes, and therefore the result still holds.

The information about other players' efforts is important both qualitatively and quantitatively. The sequential contest with five players ensures a higher total effort than the simultaneous contest with 24 players. The differences become even larger with larger contests. For example, a contest with 14 sequential players achieves higher total effort than a contest with 16,000 simultaneous players. Therefore, the information about other players' efforts is at least as important as other characteristics of the model, such as the number of players.

I also generalize the first-mover advantage result by \cite{dixit_strategic_1987}, who showed that a player who pre-commits chooses a greater level of effort and obtains a higher payoff than his followers. This leader exploits two advantages: he moves earlier and has no direct competitors. With the characterization result, I can further explore this question and compare players' payoffs and effort levels in an arbitrary sequential contest. I show that there is a strict \emph{earlier-mover advantage}---earlier players choose greater efforts and obtain higher payoffs than later players.

Finally, I provide insights for large contests. I derive a convenient approximation result in which the equilibrium efforts are directly computed using a simple formula. This result allows me to show that as the number of players becomes large, the total effort converges to the prize's value regardless of the contest structure. However, the speed of convergence to this competitive outcome or full rent-dissipation level is different under different disclosure policies. In simultaneous contests the rate of convergence is linear, whereas in sequential contests it is exponential.

My results have applications in various branches of economics, including oligopoly theory, contestability, rent-seeking, research and development, and public goods provision. For example, they provide a natural foundation for the contestability theory: if firms enter the market sequentially, the market could be highly concentrated but close to the competitive equilibrium in terms of equilibrium outcomes. The early movers would produce most of the output, but the late-movers would behave as an endogenous competitive fringe by being ready to produce more as soon as earlier players attempt to abuse their market power. Similarly, my results explain a paradox in the rent-seeking literature: explaining rent dissipation with strategic agents requires the assumption of an unrealistically large number of players. My results show that in a sequential rent-seeking contest, only a small number of active players is sufficient to achieve almost full rent dissipation.

\paragraph{Literature:}
The simultaneous version of the model has been studied extensively, starting from \cite{cournot_mathematiques_1838}. The literature on Tullock contests was initiated by \cite{tullock_welfare_1967,tullock_social_1974} and motivated by rent-seeking \citep{krueger_political_1974}.\footnote{See \cite{nitzan_modelling_1994,konrad_strategy_2009,vojnovic_contest_2015} for literature reviews on contests.} The most general treatment of simultaneous contests is provided by the literature on aggregative games \citep{selten_preispolitik_1970,acemoglu_aggregate_2013,jensen_aggregative_2018}. My model is an aggregative game only in the simultaneous case (see \cref{A:aggregative} for details).

The only sequential contest that has been studied extensively is the first-mover contest. It was introduced by \cite{von_stackelberg_marktform_1934}, who studied quantity leadership in an oligopoly.  \cite{dixit_strategic_1987} showed that there is a first-mover advantage in contests. Relatively little is known about contests with more than two periods. The only paper prior to this that studied sequential Tullock contests with more than two periods is \cite{glazer_sequential_2000}, who characterized the equilibrium in the sequential three-player Tullock contest.\footnote{\cite{kahana_sequential_2018} is an independent and concurrent contribution that provides equilibrium characterization for any fully sequential $n$-player Tullock contest. In contrast to my paper, they do not study contests with multiple players moving at once, and their approach is not applicable to other models.
} The only class of contests where equilibria are fully characterized for sequential contests are oligopolies with linear demand.\footnote{
\cite{daughety_beneficial_1990} used such a model to show that an oligopoly where players are divided between two periods is more concentrated but also closer to competitive equilibrium than an oligopoly where all players move at once. \cite{hinnosaar_stackelberg_2020} provides a literature review and shows that the linear oligopoly model has unique properties that fail when the demand is not linear.
}

More is known about large contests. Perfect competition (Marshall equilibrium) is a standard assumption in economics, and it is a baseline with which to understand its foundations. \cite{novshek_cournot_1980} showed that Cournot equilibrium exists in large markets and converges to the Marshall equilibrium. \cite{robson_stackelberg_1990} provided further foundations for Marshall equilibrium by proving an analogous result for large sequential oligopolistic markets. In this paper, I take an alternative approach. Under stronger assumptions about payoffs, I provide a full characterization of equilibria with any number of players and any disclosure structure, including simultaneous and sequential contests as opposite extremes. This allows me not only to show that the large contest limit is the Marshall equilibrium but also to study the rates of convergence under any contest structure.

The paper also contributes to the contest design literature. Previous papers on contest design include \citet{taylor_digging_1995}, \citet{che_optimal_2003}, \citet{moldovanu_optimal_2001,moldovanu_contest_2006}, and \citet{olszewski_large_2016}, which have focused on contests with private information. \cite{halac_contests_2017} studied contest design in the presence of informational externalities when players learn about the feasibility of the project. In this paper I study contest design on a different dimension: how to optimally disclose other players' efforts, when players move sequentially, in order to minimize or maximize total effort.

Finally, similar connections between disclosures and subsequent actions have been found in other settings. For example, \cite{fershtman_dynamic_1991}, \cite{varian_sequential_1994}, and \cite{wirl_dynamic_1996} used a model of dynamic voluntary public goods provision to show that if contributions are adjusted after observing earlier contributions, this may increase the free-riding problem. \cite{admati_joint_1991} and \cite{bonatti_collaborating_2011} showed similar effects in dynamic team production problems. While the driving forces in these papers are similar to the discouragement effect studied herein, none of these works addressed higher-order effects and their implications for resulting equilibria.

The paper also helps to explain empirical findings. For example, there is widespread empirical evidence of earlier-mover advantage in consumer goods markets. According to a survey by  \cite{kalyanaram_order_1995}, there is a negative relationship between a brand's entry time and brand's market share in many mature markets, including pharmaceutical products, investment banks, semiconductors, and drilling rigs. For example, \cite{bronnenberg_brand_2009} studied brands of typical consumer packaged goods and found a significant early entry advantage. The advantage is strong enough to drive the rank order of market shares in most cities.
\cite{lemus_dynamic_2021} used observational data and a lab experiment to study the impact of public leaderboards in prediction contests. They found that public leaderboards encouraged some players and discouraged others, but the overall effect was positive, improving the prediction contest's quality.

\section{Model} \label{S:model}

There are $n$ identical players $\players = \{1,\dots,n\}$ who arrive to the contest sequentially and make effort choices on arrival. At $T-1$ points in time, the sum of efforts by previous players is publicly disclosed. These disclosures partition players into $T$ groups, denoted by $\prts = (\prt_1,\dots,\prt_T)$. In particular, all players in $\prt_1$ arrive before the first disclosure and therefore have no information about other players' efforts. All players in $\prt_t$ arrive between disclosures $t-1$ and $t$ and therefore have exactly the same information: they observe the total effort of players arriving prior to disclosure $t-1$. I refer to the time interval in which players in the group $\prt_t$ arrived as period $t$. As all players are identical, the disclosure rule of the contest is fully described by the vector $\bn = (n_1,\dots,n_T)$, where $n_t = |\prt_t|$ is the number of players arriving in period $t$.\footnote{Equivalently the model can be stated as follows: $n$ players are divided across $T$ time periods, either exogenously or by the contest designer.}

Each player $i$ chooses an individual effort $x_i \geq 0$ at the time of arrival. I denote the profile of effort choices by $\bx = (x_1,\dots,x_n)$, the total effort in the contest by $X=\sum_{i=1}^n x_i$, and the cumulative effort after period $t$ by $X_t = \sum_{s=1}^t \sum_{i \in \prt_s} x_i$.
By construction, the cumulative effort before the contest is $X_0 = 0$, and the cumulative effort after period $T$ is the total effort exerted during the contest, i.e., $X_T = X$. \Cref{F:diagram_logic_0_notation} illustrates the notation with an example of the four-period contest $\bn = (3,1,1,2)$.
\begin{figure}[!ht]
    \centering
    \includegraphics[trim={0 8pt 0 8pt},clip,width=0.95\linewidth]{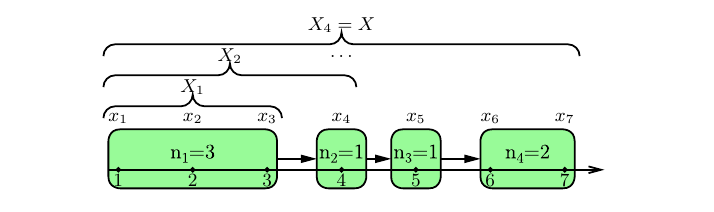}
    \caption{A contest with $7$ players and $3$ disclosures. Players 1 to 3 choose efforts $x_1$, $x_2$, and $x_3$ independently; player 4 observes $X_1=x_1+x_2+x_3$, player 5 observes $X_2=X_1+x_4$, and players 6 and 7 observe $X_3 = X_2+x_5$.}
    \label{F:diagram_logic_0_notation}
\end{figure}

Players compete for a prize of size one, the probability of winning is proportional to the level of effort, and the marginal cost of effort is one. I therefore assume the \emph{normalized Tullock payoffs}, with
\begin{equation} \label{E:utullock}
  u_i(\bx) = \frac{x_i}{X} - x_i.
\end{equation}
The model can also be interpreted as an oligopoly with unit-elastic inverse demand function $P(X)=\frac{1}{X}$ and marginal cost $c=1$, or a public goods provision (or tragedy of the commons) with the marginal benefit of private consumption $MB(X)=\frac{1-X}{X}$. 

I study pure-strategy subgame-perfect equilibria, a natural equilibrium concept in this setting: there is no private information, and earlier arrivals can be interpreted as having greater commitment power. I show that there always exists a unique equilibrium.
Throughout the paper, I maintain a few assumptions that simplify the analysis.
First, there is no private information. Second, the arrival times and the disclosure rules are fixed and common knowledge. Third, each player makes an effort choice just once upon arrival. Fourth, disclosures make cumulative efforts public. In \cref{S:discussion}, I discuss the extent to which the results rely on each of these assumptions and explain how the results extend to more general sequential games.

\section{Example} \label{S:example}

The standard Tullock contest has $n$ identical players who make their choices in isolation. Each player $i$ chooses effort $x_i$ to maximize the payoff \eqref{E:utullock}. The optimal efforts have to satisfy the first-order condition
\begin{equation}
  \frac{1}{X} - \frac{x_i}{X^2} - 1=0.
\end{equation}
Combining the optimality conditions leads to a total equilibrium effort $X^* = \frac{n-1}{n}$ and individual efforts $x_i^* = \frac{n-1}{n^2}$. The equilibrium is unique, easy to compute, and easy to generalize in various directions, which may explain the widespread use of this model in various branches of economics.

\subsection{The Problem With Standard Backward Induction} \label{SS:problem}

Consider next a three-player version of the same contest, but the players arrive sequentially and their efforts are instantly publicly disclosed. That is, players 1, 2, and 3 make their choices $x_1, x_2$, and $x_3$ after observing the efforts of previous players. I will first try to find equilibria using the standard backward-induction approach.

Player $3$ observes the total effort of the previous two players, $X_2 = x_1+x_2 < 1$ and maximizes the payoff.
The optimality condition for player $3$ is
\begin{equation} \label{E:x3X2_direct3}
\frac{1}{X_2+x_3} - \frac{x_3}{(X_2+x_3)^2} - 1 = 0
\;\;\;
\Rightarrow
\;\;\;x_3^*(X_2) = \sqrt{X_2} - X_2.
\end{equation}
Note that histories where the cumulative effort reaches one, i.e., $X_2 \geq 1$ or $x_1 \geq 1$, imply non-positive payoffs for some players and this cannot happen in equilibrium. Now, player $2$ observes $x_1<1$ and knows $x_3^*(X_2)$, and therefore maximizes
\begin{equation*}
\max_{x_2 \geq 0} \frac{x_2}{x_1+x_2+x_3^*(x_1+x_2)} - x_2
=
\max_{x_2 \geq 0} \frac{x_2}{\sqrt{x_1+x_2}} - x_2.
\end{equation*}
The optimality condition for player 2 is
\begin{equation*}
  \frac{1}{\sqrt{x_1+x_2}} -\frac{x_2}{2(x_1+x_2)^{\frac{3}{2}}} - 1 = 0.
\end{equation*}
For each $x_1 \in [0,1)$, this equation defines a unique best-response,
\begin{equation} \label{E:x2x1_direct3}
x_2^*(x_1)
=
\frac{1}{12}-x_1
+
\frac{
    \left(
    8 \sqrt{27 x_1^3 (27 x_1+1)} + 216 x_1^2+36 x_1+1
    \right)^{\frac{2}{3}}
    +24 x_1+1
}{
    12 \left(
    8 \sqrt{ 27 x_1^3 (27 x_1+1)} +216 x_1^2+36 x_1+1
    \right)^{\frac{1}{3}}
}.
\end{equation}

Finally, player $1$'s problem is
\begin{equation*}
  \max_{x_1 \geq 0} \frac{x_1}{x_1+x_2^*(x_1) +
  x_3^*(x_1+x_2^*(x_1))} - x_1,
\end{equation*}
where $x_2^*(x_1)$ and $x_3^*(X_2)$ are defined by \cref{E:x3X2_direct3,E:x2x1_direct3}.
Although the problem is not complex, it is not tractable.
Moreover, the direct approach is not generalizable for an arbitrary number of players. In fact, the best response function does not have an explicit representation for contests with a larger number of periods.

\subsection{Inverted Best-Response Approach} \label{SS:IBR}

In this paper I introduce a different approach. Instead of characterizing individual best-responses $x_i^*(X_{t-1})$, or the total efforts induced by $X_{t-1}$, i.e., $X^*(X_{t-1})$, I characterize the inverse of this relationship. For any level of total effort $X$, the inverted best-response function $f_{t-1}(X)$ specifies the cumulative effort $X_{t-1}$ prior to period $t$, that is consistent with total effort being $X$, given that the players in periods $t,\dots,T$ behave optimally.

To see how the characterization works, consider the three-player sequential contest again. In the last period, player $3$ observes $X_2$ and chooses $x_3$. Equivalently we can think of his problem as choosing the total effort $X \geq X_2$ by setting $x_3 = X-X_2$, i.e.,
\begin{equation*}
\max_{X \geq X_2} \frac{X-X_2}{X} - (X-X_2)
\;\;\;
\Rightarrow
\;\;\;
\frac{1}{X}
-\frac{X-X_2}{X^2} - 1
=
\frac{X_2}{X^2} - 1
= 0
,
\end{equation*}
which implies $X_2 = X^2$. That is, if the total effort in the contest is $X$, then before player $3$'s action, the cumulative effort had to be $f_2(X)=X^2$; otherwise, player $3$ would not be behaving optimally.

We can now think of player $2$'s problem as choosing $X \geq X_1=x_1$, which he can induce by making sure that the cumulative effort after his move is $X_2 = f_2(X)$, setting $x_2 = f_2(X) - X_1$. Therefore, his maximization problem can be written as
\begin{equation} \label{E:example_foc2}
  \max_{X \geq X_1} \frac{f_2(X)-X_1}{X} - (f_2(X)-X_1)
  \;\;\Rightarrow\;\;
\frac{f_2'(X)}{X}
-\frac{f_2(X)-X_1}{X^2} - f_2'(X) = 0.
\end{equation}
This is the key equation to examine in order to see the advantage of the inverted best-response approach. \Cref{E:example_foc2} is nonlinear in $X$ and therefore in $x_2$, which causes the difficulty for the standard backward-induction approach.
Solving this equation every period for the best-response function leads to complex expressions, and the complexity increases with each step of the recursion.
However, \eqref{E:example_foc2} is linear in $X_1$, making it easy to derive the inverted best-response function
\begin{equation*}
  f_1(X) = X_1 = f_2(X) - f_2'(X) X(1-X) = X^2 (2X -1).
\end{equation*}
The condition $X_1=f_1(X)$ aggregates the two necessary conditions of equilibrium into one, by capturing the best responses of players $2$ and $3$. It simply states that if the total effort at the end of the contest is $X$, then the cumulative effort $X_1$ had to be $f_1(X)$ after player $1$. Otherwise either player $2$ or player $3$ is not behaving optimally.

Note that $X < \frac{1}{2}$ cannot be induced by any $x_1$, as even if $x_1 = 0$, the total effort chosen by players $2$ and $3$ would be $\frac{1}{2}$. Inducing total effort below $\frac{1}{2}$ would require player 1 to exert negative effort, which is not possible. Therefore, $f_1(X)$ is defined over the domain
$\left[\frac{1}{2},1\right]$, and it is strictly increasing in this interval.

Player $1$ knows that he can induce total effort $X \geq \frac{1}{2}$ by choosing effort $x_1=f_1(X)$. Therefore we can write his maximization problem as
\begin{equation*}
\max_{X \geq \frac{1}{2}} \frac{f_1(X)}{X} - f_1(X)
\;\;\;
\Rightarrow
\;\;\;
\frac{f_1'(X)}{X}- \frac{f_1(X)}{X^2} - f_1'(X) = 0,
\end{equation*}
which implies
\begin{equation} \label{E:f0_indirect3}
0 = f_0(X) = f_1(X) - f_1'(X) X (1-X) = X^2 (6X^2 - 6X -1).
\end{equation}
This equation has again a simple interpretation---the total equilibrium quantity $X^*$ must be consistent with optimal behavior of all three players, which is now captured by the function $f_0(X)$, and with the fact that before any player chooses the action, the cumulative effort is $X_0 = 0$. Equation \cref{E:f0_indirect3} gives three candidates for the total equilibrium effort $X^*$. It is either $0$, $\frac{1}{2}-\frac{1}{2 \sqrt{3}} < \frac{1}{2}$, or $\frac{1}{2}+\frac{1}{2 \sqrt{3}} > \frac{1}{2}$. Only the highest root $X^* = \frac{1}{2}+\frac{1}{2 \sqrt{3}} \approx 0.7887$ constitutes an equilibrium.

The advantage of the inverted best-response approach is that instead of finding solutions to nonlinear equations that become increasingly complex with each recursion, this method combines all of the first-order necessary conditions into a single equation that is then solved only once.

There is a simple recursive dependence in each period that determines how the inverted best-response function evolves. At the end of the contest, i.e., after period $3$, the total effort is $f_3(X)=X$. In each of the previous periods it is equal to $f_{t-1}(X) = f_t(X) - f_t'(X) X (1-X)$. Extending the analysis from three sequential players to four or more sequential players is straightforward. It requires applying the same rule more times and solving a somewhat more complex equation at the end.

\section{Characterization} \label{S:characterization}

The main technical result of the paper is the characterization theorem (\cref{T:characterization}). It shows that each contest $\bn=(n_1,\dots,n_T)$ has a unique equilibrium and characterizes it using the inverted best-response functions $f_0,\dots,f_T$. These functions are recursively defined according to the same rule as in the previous example. The function $f_T(X)$ specifies the cumulative effort after the last period $T$ that is consistent with total effort $X$, which is clearly $f_T(X)=X$. In previous periods, the function satisfies the same rule
\begin{equation} \label{E:ftrecdef}
  f_{t-1}(X)
  = f_t(X) - n_t f_t'(X) X (1-X),
  \;\;\;
  \forall t \in \{1,\dots,T\},
  \text{where }f_T(X)=X.
\end{equation}

The only difference with the example is the term $n_t$. If there are multiple players in period $t$, then each of them has only a fractional impact on the followers' optimal responses. This means that the effect on inverted best-responses is multiplied by $n_t$.

\begin{theorem}[Characterization Theorem] \label{T:characterization}
  Each contest $\bn$ has a unique equilibrium.
  The equilibrium strategy of player $i$ in period $t$ is\footnote{Function $f_t^{-1}$ is the inverse of $f_t(X)$ in the interval $[\uX_t,1]$, where $\uX_t$ is the highest root of $f_t$. The proof of the theorem shows that $\uX_t<1$ and $f_t(X)$ is strictly increasing this interval.}
  \begin{equation} \label{E:brfunction}
    x_i^*(X_{t-1})
    =
    \begin{cases}
    \frac{1}{n_t} \left[ f_t(f_{t-1}^{-1}(X_{t-1})) - X_{t-1} \right]
    & \forall X_{t-1} < 1, \\
    0 & \forall X_{t-1} \geq 1.
    \end{cases}
  \end{equation}
  In particular, the total equilibrium effort $X^*$ is the highest root of $f_0(X)=0$, and the equilibrium effort of player $i \in \prt_t$ is $x_i^* =  \frac{1}{n_t} \left[ f_t(X^*) - f_{t-1}(X^*) \right]$.
\end{theorem}

The proof in \cref{A:proofs} starts by showing that the polynomials $f_t$ have some helpful properties. Let $\uX_t$ be the highest root of $f_t(X)=0$. These highest roots are ordered according to $t$ as $0 = \uX_T < \uX_{T-1} < \dots < \uX_1 < \uX_0$. Moreover, for all $X \in [\uX_t,\uX_{t-1})$ the function $f_{t-1}(X) < 0$ and for all $X \in [\uX_t,1]$ the function $f_t'(X)>0$. Then, the arguments in the previous section imply that there are two types of necessary conditions for equilibria. First, total equilibrium effort $X^*$ must be consistent with cumulative effort before the contest being zero, i.e., $X^*$ must be a root $f_0(X)=0$. Second, cumulative effort cannot decrease, i.e., $f_t(X^*) \geq f_{t-1}(X^*)$ for all $t$. These conditions together with the properties with $f_t$ functions imply that the highest root of $f_0(X^*)$ is the only candidate for equilibrium. Finally, assuming the followers behave according to the strategies characterized by $f_t$ functions, each player has a unique interior optimum for each $X_t < 1$, which means that the candidate for equilibrium determined above is indeed an equilibrium.

\section{Information and Effort} \label{S:information}

In this section, I show that information increases total effort in sequential contests. Before the formal result, let me give an example. Consider contests $\bn = (1,2,1)$ and $\hbn=(1,1,1,1)$. The second contest $\hbn$ is more informative as the added disclosure after player $2$ creates a finer partition of players. Direct application of \cref{T:characterization} gives equilibrium efforts $X^* = (7 + \sqrt{13})/12 < \hX^* = (6+\sqrt{24})/12$, i.e., total equilibrium effort in more informative contest is strictly higher.

The intuition for this ranking is the following. While efforts could be strategic complements or strategic substitutes, in equilibrium the efforts are high enough to make the individual efforts strategic substitutes. Compared to contest $ \bn $, the additional disclosure of player $2$'s effort in contest $\hbn$ gives player $2$ a new reason to increase his effort: it discourages player $3$. Therefore we would expect the effort of player $2$ to be higher and the effort of player $3$ to be lower than in contest $\bn$.

The remaining question is: which of these two effects is larger? The payoff function of player $2$ is $u_2(\bx) = x_2 \left( \frac{1}{X}-1\right)$, which is strictly increasing in player's own effort $x_2$ and strictly decreasing in total effort $X$. If player $2$ could increase his effort $x_2$ in a way that the discouragement effect is so large that total effort decreases, player $2$ would happily exploit this opportunity, and the outcome would not be an equilibrium. Therefore, the discouragement effect is less than one-to-one, implying that the total effort is increased.

The full comparison of the two contests must also consider how players $1$ and $4$ respond to the change of game conditions. Their incentives are driven by indirect effects. For example, player $1$ may want to influence player $2$ to exert more or less effort in the more informative contest, depending on how this second-order impact affects other players.

Capturing the indirect effects requires some new notation. Let me use a contest $\bn =(1,2,1)$ again to illustrate the construction of relevant variables. All four players in this contest observe their own efforts (regardless of the disclosure rule), and the number of players clearly affects the outcomes of the contest. I call this the first level of information and denote it as $S_1=n=4$. More importantly, some players directly observe the efforts of some other players. Players $2$ and $3$ observe the effort of player $1$ and player $4$ observes the efforts of all three previous players. Therefore, there are five direct observations of other players' effort levels. I call this the second level of information and denote it as $S_2=5$. Finally, player $4$ observes players $2$ and $3$ observing player $1$. There are two indirect observations of this kind, which I call the third level of information and denote as $S_3=2$. In contests with more periods, there would be more levels of information---observations of observations of observations and so on.

I call a vector $\bS(\bn) = (S_1(\bn),\dots,S_T(\bn))$ the measure of information in a contest $\bn$. In the example described above, $\bS(1,2,1)=(4,5,2)$. Formally, $S_k(\bn)$ is the sum of all products of $k$-combinations of set $\{n_1,\dots,n_T\}$. For example, in a sequential $n$-player contest $\bn=(1,1,\dots,1)$, $S_k(\bn)$ is simply the number of all $k$-combinations, i.e., $S_k(\bn) = \frac{n!}{k! (n-k)!}.$
With this notation, I can now state the second main result.
\begin{theorem}[Information Theorem] \label{T:information}
  Total effort in contest $\bn$ is a strictly increasing function $X^*(\bS(\bn))$.
\end{theorem}

There are two key steps in the proof. The first step shows by induction that the inverted best-response functions can be expressed using the measures of information as:
\begin{equation} \label{E:ft_withSg}
  f_t(X) = X - \sum_{k=1}^{T-t} S_k(\bn^t) g_k(X),
\end{equation}
where $\bn^t=(n_{t+1},\dots,T)$ is a vector of integers describing the sub-contest that starts after period $t$ and $\mathbf{S}(\bn^t)$ denotes its measures of information, i.e., $S_k(\bn^t)$ is the sum of all products of $k$-combinations of $\bn^t$. The functions $g_k(X)$ are defined recursively and independently of the contest $\bn$ as $g_1(X)=X (1-X)$, and for all $k>0$, $g_{k+1}(X)=-g_k'(X) X (1-X)$.
In particular, $f_0(X)$ takes the following form:
\begin{equation} \label{E:f0_withSg}
  f_0(X) = X - \sum_{k=1}^{T} S_k(\bn) g_k(X),
\end{equation}
Remember, that total equilibrium effort $X^*$ is the highest root of $f_0(X)$ and this function is strictly increasing above its highest root. The second key step of the proof shows that functions $g_k(X^*)>0$ at the equilibrium value $X^*$. Now, increasing $\bS(\bn)$ decreases the value of $f_0(X^*)$ at the original equilibrium value. Therefore the highest root of the new function $f_0(X)$ must be higher than the original one.

\Cref{E:f0_withSg} also sheds some light on the reason for this result. The positive weights on the measures of information, $g_k(X^*)$, can be interpreted as higher-order strategic substitutability terms. In particular, $g_1(X)$ captures the concavity of payoff functions---as agent $i$ increases his effort, the incentive to increase the effort further decreases and is positive for all $X$. The term $g_2(X)$ captures the standard strategic substitutability. 
In the example above, if player $2$ increases effort, player $3$ who observes this deviation has an incentive to decrease effort. The next term $g_3(X)$ captures a higher-order indirect incentive: 
player $4$, who observes the response of player $3$ has an incentive to decrease effort as well (beyond the direct effect of responding to player $2$). The fact that near equilibrium $g_k(X^*)>0$ for all $k$, means that all these effects move the equilibrium outcomes to the same direction---in more informative contests (in the sense of $\bS(\bn)$) earlier-movers exert more effort, later-movers less effort, but the total equilibrium effort is higher.

\Cref{T:information} defines a partial order on all contests---if contests can be ranked about the measures of information $\bS(\hbn)>\bS(\bn)$, i.e., $S_k(\hbn)\geq S_k(\bn)$ for all $k$ and the inequality is strict at least for one $k$, then $X^*(\bS(\hbn))>X^*(\bS(\bn))$. In the example above, increasing informativeness in contests by adding public disclosures increases vector $\bS$ and therefore total effort, or more concretely $\bS(1,2,1)=(4,5,2) < \bS(1,1,1,1)=(4,6,4,1)$.
To complete the order, we would have to know how to weigh different measures of information. \Cref{E:f0_withSg} shows that correct weights are $g_k(X^*)$; i.e., by magnitudes of discouragement effects near equilibrium. The following lemma shows that lower information measures have a higher weight.
\begin{lemma}[Decreasing Weights] \label{L:gktullock}
    $g_{k-1}(X^*)>g_k(X^*)$ for each $k \geq 2$.
\end{lemma}

While direct effects have a larger impact than the indirect ones, the indirect effects are not qualitatively unimportant. Compare for example two seven-player contests $\bn=(3,4)$ and $\hbn=(1,1,5)$. The first contest $\bn$ has $12$ direct observations whereas $\hbn$ has only $11$. Nevertheless, the total effort in $\bn$ is lower than in contest $\hbn$. This is because of indirect effects, $S_3(\hbn)=5>0$.
Intuitively, in the contest $\hbn$ player $1$ knows that in addition to influencing all followers directly, he the five last-movers through the behavior change of player $2$. This indirect effect is missing in $\bn$.

\Cref{T:information} has several direct implications summarized by the following corollary.
\begin{corollary}[Implications of the Information Theorem] \label{C:information}
  Take two contests $\bn^1$ and $\bn^2$, with corresponding partitions $\prts^1$ and $\prts^2$, and let $X^1=X^*(S(\bn_1))$ and $X^2=X^*(S(\bn_2))$ be the corresponding total equilibrium efforts.
  \begin{enumerate}
      \item Comparative statics of $\bn$: if $\bn^1 < \bn^2$,  then $X^1 < X^2$. This includes the case when $n_t^1=0 < \hn_t^2$ for some $t$, i.e., $\bn^2$ has more periods.
      \item Independence of permutations: if $\bn^1$ is a permutation of $\bn^2$, then $X^1 = X^2$.
      \item Disclosures increase effort: if $\prts^1$ is a coarser partition than $\prts^2$, then $X^1<X^2$.
      \item Homogeneity increases effort: if $\sum_t n_t^1 = \sum_t n_t^2$ and there exist $t,t'$ such that $n_t^1 n_{t'}^1 < n_t^2 n_{t'}^2$ and $n_s^1 = n_s^2$ for all $s \neq t,t'$, then $X^1 < X^2$.
      \item For any $n$, the simultaneous contest $\bn = (n)$ minimizes total effort, and the fully sequential contest $\bn = (1,1,\dots,1)$ maximizes the total effort. For fixed number of periods $T$, contests that allocate players into groups that are as equal as possible maximize the total effort.
  \end{enumerate}
\end{corollary}

The first implication is intuitive, adding players to any period or adding periods to any contest increases the total effort. Note, however, that this does not imply that the total effort increases with the total number of players. In the examples above, total effort in the three-player sequential contest was $0.7887$, whereas in the four-player simultaneous contest, it was $0.75$. The second implication is more surprising---reallocating disclosures in a way that creates a permutation of $\bn$ does not affect the total effort. For example, a first-mover contest $(1,n-1)$ gives the same total effort as the last-mover contest $(n-1,1)$. The third implication that disclosures increase effort was already discussed above. The fourth implication gives even clearer implications for the optimal contest. Namely, more homogeneous contests give higher total effort. Intuitively, a contest is more homogeneous if players are divided more evenly across periods. For example, a contest $\hbn=(2,2)$ is more homogeneous than $\bn = (1,3)$ and also has more direct observations of efforts as $2 \times 2 = 4 > 3 = 1 \times 3$.

Therefore, if the goal is to minimize the total effort (such as in rent-seeking contests), then the optimal policy is to minimize the available information, which is achieved by a simultaneous contest. Transparency gives earlier players incentives to increase efforts to discourage later players, but this discouragement effect is less than one-to-one and therefore increases total effort. On the other hand, if the goal is to maximize the total effort (such as in research and development), then the optimal contest is fully sequential as it maximizes the incentives to increase efforts through this discouragement effect. If the number of possible disclosures is limited (for example, collecting or announcing information is costly), then it is better to spread the disclosures as evenly as possible.

\section{Earlier-Mover Advantage} \label{S:earlier-mover}

\cite{dixit_strategic_1987} showed that there is a first-mover advantage. If one player can pre-commit, the first-mover chooses a strictly higher effort and achieves a strictly higher payoff than the followers. Using the tools developed here, I can explore this result further. Namely, the first mover has two advantages compared to the followers. First, he moves earlier, and his action may impact the followers. Second, he does not have any direct competitors in the same period. I can now distinguish these two aspects. For example, what would happen if $n-1$ players chose simultaneously first, and the remaining player chose after observing their efforts? More generally, in an arbitrary sequence of players, which players choose the highest efforts and which ones get the highest payoffs? The answer to all such questions turns out to be unambiguous---there is a strict earlier-mover advantage.

\begin{proposition}[Earlier-Mover Advantage] \label{P:earliermover}
    The efforts and payoffs of earlier
     players are strictly higher than for later players.
\end{proposition}

The equilibrium payoff of a player $i$ is in $u_i(\bx^*) = x_i^* \left( \frac{1}{X^*}-1\right)$, and since $X^*$ is the same for all the players, payoffs are proportional to efforts. Therefore, it suffices to show that the efforts of earlier players are strictly higher. Using \cref{T:characterization} and \cref{E:ft_withSg}, I can express the difference between the equilibrium efforts of players $i$ and $j$ from consecutive periods $t$ and $t+1$ as
\begin{equation}
  x_i^* - x_j^*
  = \sum_{k=1}^{T-t} \left[
    S_k(\bn^t)
    -
    S_k(\bn^{t+1})
  \right] g_{k+1}(X^*)
\end{equation}
where $\bn^{t+1} = (n_{t+2},\dots,n_T)$ is the sub-contest starting after period $t+1$ and $\bn^t = (n_{t+1},\bn^{t+1})$ is the sub-contest starting after period $t$. Clearly, $S_k(\bn^t) > S_k(\bn^{t+1})$ for all $k$; i.e., there is more information on all levels in a strictly longer contest. As $g_{k+1}(X^*)>0$ for each $k$ the whole sum is strictly positive. The intuition of the result is straightforward: players in earlier periods are observed by strictly more followers than the players from the later periods. Therefore, in addition to the incentives that later players have, the earlier players have an additional incentive to exert more effort to discourage later players.

\section{Large Contests} \label{S:largecontests}

Numeric comparison of simultaneous and sequential contests highlights that the information about other players' efforts is at least as an important factor in determining the total effort as other parameters, such as the number of players. For example, the total effort in the simultaneous contest with ten players is $0.9$, whereas the total effort with four sequential players is $0.9082$. A fifth sequential player increases the total effort to $0.9587$. A simultaneous contest with the same total effort requires $24$ players. \Cref{F:simseqequiv} shows that the comparison becomes even more favorable for sequential contests with large $n$.
\begin{figure}[!ht]
    \centering
    \includegraphics[trim={0 21pt 0 18pt},clip,width=0.6\linewidth]{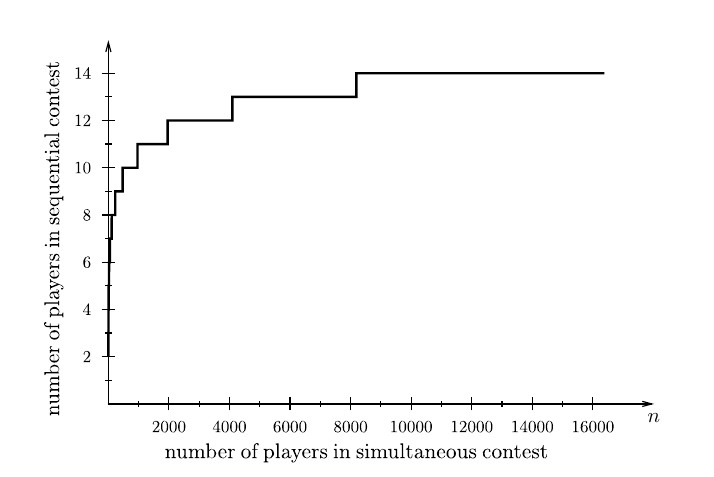}
    \caption{Number of players in a sequential contest that leads to the same total effort as a simultaneous contest with $n$ players}
    \label{F:simseqequiv}
\end{figure}

The following proposition gives the reason for this connection. As the number of players becomes large, there will be full rent dissipation, i.e., total effort converges to $1$ regardless of contest structure. However, in large simultaneous contests $1-X^* \approx \frac{1}{n}$, so the convergence is linear, whereas in large sequential contests $1-X^* \approx \frac{1}{2^n}$, which means exponential convergence.

\begin{proposition}[Large Contests] \label{P:largecontests}
  Fix $T \in \N$ and a sequence of contests $(\bn^n)_{n=3}^{\infty}$, such that contest $\bn^n$ is $n$-player contest with at most $T$ periods. Let $X^n=X^*(\bS(\bn^n))$ and for each player $i$, let $x_i^n$ the equilibrium effort in contest $\bn^n$. For all $t \leq T$ and all $i \in \prt_t^n$,
  \begin{equation} \label{E:largecontests}
    \lim_{n \to \infty} \left[
    X^n - \left( 1-\frac{1}{\prod_{t=1}^T (1+n_t^n)} \right)
    \right]
    = 0
    \;\;
    \text{and}
    \;\;
    \lim_{n \to \infty} \left[ x_i^n - \frac{1}{\prod_{s=1}^t (1+n_s^n)} \right] = 0.
  \end{equation}
\end{proposition}

In large simultaneous contests, each player chooses negligible effort $x_i^* \approx \frac{1}{n}$. In contrast, in large sequential contest, the individual equilibrium efforts are $\bx^* \approx \left( \frac{1}{2}, \frac{1}{4},\dots,\frac{1}{2^n} \right)$. The earlier movers choose much larger efforts and therefore achieve much larger payoffs than the followers. However, if they would try to exploit their dominant position, the followers would increase their efforts.\footnote{This effect is similar to contestability theory \citep{baumol_contestable_1988}, but competitive fringe arising endogeneously from the order of arrivals.}

\section{Discussion} \label{S:discussion}

I showed that each contest $\bn$ has a unique equilibrium. It is in pure strategies and simple to compute. The total equilibrium effort is strictly increasing in information. This implies that the optimal contest for maximizing total effort is fully sequential, e.g., R\&D contests benefit from full transparency. On the other hand, if the goal is to minimize the total effort, such as rent-seeking contests, the optimal contest is non-transparent, i.e., the simultaneous contest. Further, there is a strict earlier-mover advantage: players in earlier periods exert strictly greater efforts and obtain strictly higher payoffs. Total effort converges to full dissipation linearly with the number of players in large simultaneous contests but exponentially in large sequential contests.

The results in this paper hold much more generally than the model discussed herein. 
First, I assumed that players exert efforts only at their arrival, and that their efforts are publicly observable for players in the following periods. Given that players benefit from the discouragement effect, they would not hide or delay their actions. Thus, the outcomes would be unchanged if players could take hidden actions or take actions over multiple periods. This was shown by \cite{yildirim_contests_2005} in the two-player case.

Second, all results hold for a general payoff functions in the form $u_i(\bx)=x_i h(X)$. For example, in the case of Tullock contest $h(X) = \frac{v}{X}-c$. In the case of oligopoly, $x_i$ can be quantity or capacity, and $h(X)=P(X)-c$, which is the difference between inverse demand and the marginal cost. With the tragedy of the commons or public goods provision problem, $x_i$ is private consumption, and $h(X)$ is the marginal benefit of private consumption, decreasing with the total consumption of resources (or increases with the amount of public good). Naturally, the function $h(X)$ must satisfy some technical properties---intuitively, actions must be higher-order strategic substitutes. I discuss these conditions in detail in \cref{A:xhX}, describe a class of functional forms where the sufficient conditions for the results above are satisfied, and provide some examples where the analysis may fail if the conditions are not satisfied.

Third, the assumption of linear costs is important for the simplicity of characterization but can be relaxed, at least within a certain range. It is known that even Tullock contests with nonlinear costs may not have equilibria in pure strategies.\footnote{Tullock contest with nonlinear cost function $c(x_i)=x_i^r$ has a pure-strategy equilibrium if and only if $r \geq \frac{n-1}{n}$, i.e., when the cost function is either convex or mildly concave \citep{baye_solution_1994}. In other settings, such as group contests, the pure-strategy equilibria may not exists if the cost function is concave \citep{esteban_conflict_1999,esteban_linking_2011}.} In \cref{A:quadratic}, I show that the inverted best-response approach is more complicated, but still tractable when costs are quadratic, $c(x_i)=x_i+\beta x_i^2$. All results except independence of permutations still hold for a sufficiently small positive or negative $\beta$.

Finally, I assumed that agents are identical in all dimensions, except that players in later groups observe and respond to players' efforts in earlier groups. In \cref{A:asymmetric}, I show that the inverted best-response approach can be applied to players with heterogeneous payoffs. However, there is a new complication: players may find it optimal to stay inactive at different thresholds. This means that earlier-movers may sometimes find it optimal to deter entry by followers and the order of players becomes an important determinant of outcomes.
\cite{hinnosaar_price_2019} extends the methodology to another type of player heterogeneity, where the game is played on a network. Players only observe the choices of players they are linked to. This analysis shows that there is a connection between weighted measures of information and standard centrality measures from network theory.

\bibliography{zotero_contests}
\bibliographystyle{chicago}

\clearpage
\clearpage
\renewcommand{\appendixpagename}{Appendix}
\appendixpage
\appendix

\section{Proofs} \label{A:proofs}

\subsection{Proof of the Characterization Theorem (\Cref{T:characterization})} \label{A:characterization}

Before proving \cref{T:characterization}, it is useful to define the following property.\footnote{
	I am calling this property an assumption here because the following analysis applies to any payoff functions, for which $f_t$ functions defined by \cref{E:ftrecdef} satisfy this property. In the case of normalized Tullock payoffs, the assumption is always satisfied. I show this in \cref{P:cond1_tullock}.
}
	
\begin{assumption}[Inverted best responses are well-behaved] \label{C:cond1}
	Clearly, $f_T(X)=X$ has unique root $\uX_T=0$.
	For all $t=0,\dots,T-1$, the function $f_t$ has the following properties:
	\begin{enumerate}
		\item $f_t(X)=0$ has a root in $[\uX_{t+1},1]$. Let $\uX_t$ be the
		highest such root.
		\item $f_t(X)<0$ for all $X \in [\uX_{t+1},\uX_t)$.
		\item $f_t'(X)>0$ for all $X \in [\uX_t,1]$.
	\end{enumerate}
	Moreover, $\uX_0 \in (0,1)$.
\end{assumption}	
	
The proof of \cref{T:characterization} has two parts. The first part is \cref{P:cond1_tullock} in \cref{A:tullock_cond1} that shows that $f_t$ functions satisfy \cref{C:cond1}. The proof relies on keeping track of the roots of $f_t$ functions. The second part in \cref{A:cond1_characterization} establishes the theorem's claims. Briefly, it shows that a behavior where each player $i$ in each period $t$ behaves according to \cref{E:brfunction} and expects that total effort induced by cumulative effort $X_t$ to be $f_t^{-1}(X)$, is an equilibrium and in fact it is the only equilibrium. The proof is divided to five lemmas:
\begin{enumerate}
	\item \Cref{L:characterization_01} shows that in all histories where $X_{t-1}<1$, each player in period $t$ chooses strictly positive effort, but these added efforts in period $t$ are small enough so that the cumulative effort after period $t$ remains strictly below one, $X_t < 1$. On the other hand, in histories where $X_{t-1} \geq 1$, the players in period $t$ exert no effort.
	Therefore, on the equilibrium path $X_t < 1$ for all $t$.
	\item \Cref{L:characterization_necc} shows that $X_t = f_t(X)$ is a necessary condition for equilibrium. In particular, $f_0(X)=0$ is a necessary condition for equilibrium, and therefore $X^*$ must be a root of $f_0(X)$.
	\item \Cref{L:characterization_welldefined} shows that under \cref{C:cond1}, the inverse function $f_{t-1}^{-1}(X_{t-1})$ is well-defined and strictly increasing, $f_{t-1}^{-1}(0)=\uX_{t-1}$ and $f_{t-1}^{-1}(1)=1$.
	\item \Cref{L:characterization_brs} shows that the best-response function of player $i \in \prt_t$ after cumulative effort $X_{t-1}$ is $x_i^*(X_{t-1}) = \frac{1}{n_t} \left[ f_t(f_{t-1}^{-1}(X_{t-1})) - X_{t-1} \right]$ for all $X_{t-1}<1$ and $x_i^*(X_{t-1})=0$ for all $X_{t-1}\geq 0$. On the equilibrium path the individual efforts are $x_i^* = \frac{1}{n_t} \left[ f_t(X^*) - f_{t-1}(X^*) \right]$.
	\item Finally, \cref{L:characterization_SOC} verifies that the unique candidate for equilibrium, i.e., $\bx^*$ specified in the theorem, is indeed an equilibrium, which completes the proof.
\end{enumerate}
The combination of these results proves \cref{T:characterization}.

\subsection{Proof that \Cref{C:cond1} is Satisfied (\Cref{P:cond1_tullock})} \label{A:tullock_cond1}

\begin{proposition} \label{P:cond1_tullock}
	Inverted best responses $f_0,\dots,f_T$ defined by \cref{E:ftrecdef} are well-behaved.
\end{proposition}

Before giving the proof of \cref{P:cond1_tullock}, let me briefly describe its key idea. The function $f_{t+1}$ is a polynomial of degree $r=T-t$, so it can have at most $r$ roots. By keeping track of all the roots, I show by induction that all $r$ roots are real and in $[0,1)$, with the highest being $\uX_{t+1}$. Therefore, all $r-1$ roots of the derivative $f_t'$ are also real and in $[0,\uX_{t+1})$. Evaluating $f_t$ at $\uX_{t+1}$ and $1$, we get
\begin{align*}
	f_t(\uX_{t+1})
	&= \underbrace{f_{t+1}(\uX_{t+1})}_{=0}
	- n_{t+1} \underbrace{f_{t+1}'(\uX_{t+1})}_{>0} \underbrace{\uX_{t+1} (1-\uX_{t+1})}_{>0}
	< 0 \\
	f_t(1) &= f_{t+1}(1)
	- n_{t+1} f_{t+1}'(1) \underbrace{1 (1-1)}_{=0}
	= f_{t+1}(1) = \dots = f_T(1)=1>0.
\end{align*}
This implies that $f_t$ must have a root $\uX_t \in (\uX_{t+1},1)$. Moreover, since the
highest root of its derivative is again below $\uX_t$, it is strictly
increasing in $[\uX_t,1]$. Finally, I show that the second highest root of
$f_t$ is strictly below $\uX_{t+1}$, so that $f_t(X)<0$ for all
$[\uX_{t+1},\uX_t)$. Proving this requires keeping track of all the roots.

\begin{proof}[Proof of \cref{P:cond1_tullock}]
	First note that $f_T(X)=X$ is a polynomial of degree 1, and each step of the recursion adds one degree, so $f_t(X)$ is a polynomial of degree $T+1-t$, which I denote by $r$ for brevity. The following two technical lemmas describe the values of the polynomials $f_t$ at $1$ and the number of roots at $0$.
	
	\begin{lemma}
		\label{L:ftat1}
		$f_t(1)=1$ for all $t=0,\dots,T$.
	\end{lemma}
	\begin{proof}
		$f_{t-1}(1)
		= f_t(1) - n_t f_t'(1) 1 (1-1)
		= f_t(1)
		= f_T(1) = 1$.
	\end{proof}
	
	\begin{lemma}
		\label{L:ftat0}
		$f_t(0) = 0$ for all $t=0,\dots,T$. Depending on $\bn$, there could be either one or two roots at zero:
		\begin{enumerate}
			\item If $n_s = 1$ for some $s>t$, then $f_t(X)$ has exactly two roots at zero.
			\item Otherwise, i.e., if $n_s \neq 1$ for all $s>t$, then $f_t(X)$ has exactly one root at zero.
		\end{enumerate}
	\end{lemma}
	
	\begin{proof}
		As $f_t(X)$ is a polynomial of degree $r=T+1-t$, it can be expressed as
		\begin{equation*}
			f_t(X) = \sum_{s=0}^r c_s^t X^s
			\;\;\; \Rightarrow \;\;\;
			f_t'(X) = \sum_{s=1}^r c_s^t s X^{s-1},
		\end{equation*}
		where $c_0^t,\dots,c_r^t$ are the coefficients. Therefore,
		\begin{equation*}
			f_{t-1}(X)
			=
			c_0^t
			+ c_1^t (1-n_t)
			+ \sum_{s=2}^r \left[
			c_s^t (1-s n_t)
			+ n_t c_{s-1}^t (s-1)
			\right] X^s
			+n_t c_{T+1-t}^t (T+1-t) X^{T+2-t}.
		\end{equation*}
		As $f_T(X)=X$, we have that $c_0^T = 0$ and so $c_0^t=0$ for all $t$. Therefore, each $f_t$ has at least one root at $0$.
		Next, $f_{t-1}(X)$ has two roots at zero if and only if $c_1^{t-1} = c_1^t (1-n_t)=0$. This can happen only if either $c_1^t = 0$ (i.e., $f_t(X)$ has two roots at zero) or $n_t=1$. As $f_T(X)=X$, we have that $c_1^T=1$ and therefore,  $f_t(X)$ does indeed have two roots at zero if and only if $n_s=1$ for some $s>t$.
		
		Finally, $f_{t-1}(X)$ would have three roots at zero only if $c_2^{t-1}=c_1^{t-1}=0=c_0^{t-1}$. This would require that $c_2^{t-1} = c_2^t (1-2 n_t) + n_t c_1^t = c_2^t (1-2 n_t)=0$. Since $2n_t \neq 1$, this can happen only when $c_2^t=0$. But note that $f_{T-1}(X) = n_T X^2 -(1-n_T) X$, so that $c_2^{T-1} = n_T \neq 0$. Therefore, $f_t(X)$ cannot have more than two roots at zero.
	\end{proof}
	
	\begin{lemma} \label{L:leadingcoefficient}
		The leading coefficient of $f_t$ is $(T-t)! \prod_{s=t+1}^T n_s > 0$.
	\end{lemma}
	\begin{proof}
		Using the same notation as in \cref{L:ftat0}, the leading coefficient of $f_{t-1}(X)$ is
		$
		c_{r+1}^{t-1}
		= r n_t c_r^t
		= r! \prod_{s=t}^T n_s$.
	\end{proof}
	
	Now I can proceed with the proof of \cref{P:cond1_tullock} itself. The proof uses that fact that the $f_t$ is a polynomial of degree $r=T+1-t$ and keeps track of all of its roots. In particular, it can be expressed as
	\begin{equation}
		f_t(X) = c_t \prod_{s=1}^r (X-X_{s,t}),
	\end{equation}
	where $c_t>0$ is the leading coefficient and $X_{1,t},\dots,X_{r,t}$ are the $r$ roots. By \cref{L:ftat0}, either one or two of these roots are equal to zero. I show by induction that all other roots are distinct and in $(0,1)$.
	
	Let us consider the case of a single zero root first, i.e., assume that $0=X_{1,t}<X_{2,t}<\ldots<X_{r,t}<1$. We can express the derivative of $f_t$ as
	\begin{equation*}
		f_t'(X)
		= c_t \sum_{i=1}^r \prod_{s \neq i} (X-X_{s,t}).
	\end{equation*}
	Therefore, at root $X_{j,t}$, the polynomial $f_t'(X)$ takes value
	\begin{equation}
		f_t'(X_{j,t})
		= c_t \prod_{s \neq j} (X_{j,t} - X_{s,t}).
	\end{equation}
	In particular, at the highest root, $f_t'(X_{r,t})>0$, and at the second highest $f_t'(X_{r-1,t})<0$; therefore, $f_t'$ must have a root $Y_{r-1,t} \in (X_{r-1,t},X_{r,t})$. By the same argument, there must be a root $Y_{s,t}$ of $f_t'$ between each of the two adjacent distinct roots of $f_t$. As $f_t'$ is a polynomial of degree $r-1$, this argument implies that all the roots of $f_t'$ are distinct and such that
	\begin{equation*}
		X_{1,t}=0<Y_{1,t}<X_{2,t}<Y_{2,t}<\dots<X_{r-1,t}<Y_{r-1,t}<X_{r,t}<1.
	\end{equation*}
	In particular, $\sgn f_t'(X_{s,t}) = \sgn f_t(Y_{s,t})$ for all $s \in \{1,\dots,r-1\}$.
	Next, note that $f_t(1)=1>0$ and, as the highest root of $f_t'$ is $Y_{r-1,t}<X_{r,t}$, this implies $f_t'(X_{r,t}) >0$, and so
	\begin{equation*}
		f_{t-1}(X_{r,t})
		= f_t(X_{r,t}) - n_t f_t'(X_{r,t}) X_{r,t} (1-X_{r,t})
		< 0.
	\end{equation*}
	Therefore, $f_{t-1}$ must have a root $X_{r+1,t-1} \in (X_{r,t},1)$. Now, for each $s \in \{2,r-1\}$
	\begin{equation*}
		f_{t-1}(Y_{s,t}) = f_t(Y_{s,t})
		\text{ and }
		f_{t-1}(X_{s,t}) = -n_t f_t'(X_{s,t}) X_{s,t} (1-X_{s,t}).
	\end{equation*}
	Hence, $\sgn f_{t-1}(Y_{s,t}) = \sgn f_t(Y_{s,t}) = \sgn f_t'(X_{s,t}) = -f_{t-1}(X_{s,t})$. This means that $f_{t-1}$ must have a root $X_{s+1,t-1} \in (X_{s,t},Y_{s,t})$. This argument determines $r-2$ distinct roots in $(X_{2,t},Y_{r-1,t})$. By \cref{L:ftat0}, $f_{t-1}$ also has at least one root $X_{1,t-1}=0$.
	
	We have therefore found $1+r-2-1=r$ distinct real roots of $f_{t-1}$ that is a polynomial of degree $r+1$. Thus, the final root $X_{2,t}$ must also be real. By \cref{L:ftat0}, if $n_t=1$,then the $f_{t-1}$ must have two roots at zero; so, $X_{2,t}=0$. Let us consider the remaining case where $n_t > 1$. By \cref{L:ftat0}, $X_{2,t} \neq 0$. To determine its location, consider the function $f_{t-1}^X(X) = f_{t-1}(X)/X$. Note that
	\begin{equation*}
		f_t^X(X) = \frac{f_t(X)}{X} = c_t \prod_{s>0} (X-X_{s,t})
		\Rightarrow
		f_t^X(0) = c_t \prod_{s>0} (-X_{s,t})
	\end{equation*}
	and
	\begin{equation*}
		f_t'(0) = c_t \prod_{s>0} (-X_{s,t}).
	\end{equation*}
	Therefore,
	\begin{equation*}
		f_{t-1}^X(0)
		= f_t^X(0)-n_t f_t'(0) (1-0)
		= c_t \prod_{s>0} (-X_{s,t}) \left[
		1 - n_t
		\right]
		= f_t'(0) \left[
		1 - n_t
		\right].
	\end{equation*}
	We assumed that $n_t>1$; so, $\sgn f_{t-1}^X(0) = - \sgn f_t'(0)$. Evaluating the function $\sgn f_{t-1}^X$ at $Y_{1,t}$ gives
	\begin{equation*}
		\sgn f_{t-1}^X(Y_{1,t})
		= \sgn f_t(Y_{1,t})
		= \sgn f_t'(X_{1,t})
		= - \sgn f_{t-1}^X(0).
	\end{equation*}
	Hence, $f_{t-1}^X$ must have a root $X_{2,t-1} \in (0,Y_{1,t})$. As $f_{t-1}(X)=X f_{t-1}^X(X)$, it must be a root of $f_{t-1}$ as well. We have therefore located all $r+1$ roots of $f_{t-1}$, which are all distinct in this case.
	
	Let us now get back to the case where $f_t$ had two roots at zero. By the same argument as above, there must be a root of $f_t'$ between each positive root of $f_t$. As there are $r-2$ positive roots, this determines $r-3$ distinct positive roots of $f_t'$. It is also clear that $f_t'$ must have exactly one root at zero. Polynomial $f_t'$ has $r-1$ roots, and we have determined that $r-2$ of them are real and distinct. Thus, the remaining root must be real. To determine its location, using the above approach, let $f_t^{X'}(X) = \frac{f_t'(X)}{X}$. Then as $f_t'(X_{r,t})>0$, we have $f_t^{X'}(X_{r,t})>0$. Similarly, $f_t^{X'}(X_{r-1,t})<0$, and so on. In particular, $f_t^{X'}(X_{3,t}) < 0$ if $r$ is even, and $f_t^{X'}(X_{3,t})>0$ if $r$ is odd. Now,
	\begin{equation*}
		f_t^{X'}(0) = 2 c_t \prod_{s>2} (-X_{s,t}),
	\end{equation*}
	which is strictly positive if $r$ is odd and strictly negative if $r$ is even, so that $\sgn f_t^{X'}(0) = - \sgn f_t^{X'}(X_{3,t})$. Hence, $f_t^{X'}$ must have a root $Y_{2,t} \in (0,X_{3,t})$. Clearly this $Y_{2,t}$ is also a root of $f_t'(X) = X f_t^{X'}(X)$. Now we have found all $r-1$ roots of polynomial $f_t'$ and
	\begin{equation*}
		X_{1,t}=Y_{1,t} = X_{2,t}=0<Y_{2,t} <X_{3,t}<\ldots < X_{r-1,t} < Y_{r,t} < X_{r,t}.
	\end{equation*}
	Again, $\sgn f_t'(X_{s,t})=\sgn f_t(Y_{s,t})$ for all $s \in \{2,\dots,r-1\}$.
	
	By the same arguments as above, $f_{t-1}$ has a root $X_{r+1,t-1} \in (X_{r,t},1)$ and $r-3$ roots $X_{s+1,t-1} \in (X_{s,t},Y_{s,t})$ for each $s \in \{3,r-1\}$. Also, by \cref{L:ftat0}, $f_{t-1}$ must have two roots at zero. Therefore, we have determined $1+r-3+2=r$ roots of $f_{t-1}$, and so the final root must also be real. The argument for determining this root is similar to the previous case. Let $f_{t-1}^{X^2}(X) = \frac{f_{t-1}(X)}{X^2}$. Then
	\begin{equation*}
		f_{t-1}^{X^2}(X)
		=
		f_t^{X^2}(X)
		-
		n_t f_t^{X'}(X) (1-X).
	\end{equation*}
	Therefore,
	\begin{equation*}
		f_{t-1}^{X^2}(0)
		= c_t \prod_{s>2} (-X_{s,t}) (1- 2 n_t),
	\end{equation*}
	so that $\sgn f_{t-1}^{X^2}(0) = - \sgn f_t^{X'}(0)$. Also,
	\begin{equation*}
		f_{t-1}^{X^2}(Y_{2,t})
		= f_t^{X^2}(Y_{2,t}).
	\end{equation*}
	Since $Y_{2,t} >0$ and $X_{3,t}>0=X_{2,t}$, we have that
	\begin{align*}
		\sgn f_{t-1}^{X^2}(Y_{2,t})
		&= \sgn f_t^{X^2}(Y_{2,t})
		= \sgn f_t(Y_{2,t})
		= -\sgn f_t(Y_{3,t}) \\
		&= -\sgn f_t'(X_{3,t})
		= -\sgn f_t^{X'}(X_{3,t})
		= \sgn f_t^{X'}(0)
		= -\sgn f_{t-1}^{X^2}(0).
	\end{align*}
	Therefore, $f_{t-1}^{X^2}$ must have a root in $(0,Y_{2,t})$ which must also be a root of $f_{t-1}$. Again, we have found all $r+1$ roots of $f_{t-1}$.
	
	In all cases, we found that
	\begin{enumerate}
		\item $X_{r+1,t-1}  \in (X_{r,t},1)$; i.e., indeed the highest root of $f_{t-1}$ is between the highest root of $f_t$ and $1$.
		\item $[X_{r,t},X_{r+1,t-1}) \subset (X_{r,t-1},X_{r+1,t-1})$, so that $f_{t-1}(X)<0$ for all $X \in [X_{r,t},X_{r+1,t-1})$.
		\item By the same argument as above (or by the Gauss-Lucas theorem), $X_{r+1,t-1} > Y_{r,t-1}$, so that $f_{t-1}'(X) > 0$ for all $X \in [X_{r+1,t-1},1]$.
	\end{enumerate}
\end{proof}

\subsection{Proof that \Cref{C:cond1} Implies \Cref{T:characterization}} \label{A:cond1_characterization}

\begin{lemma} \label{L:characterization_01} Depending on $X_{t-1}$, we have two cases:
  \begin{enumerate}
      \item If $X_{t-1}<1$, then $x_i > 0$ for all $i \in \prt_t$ and $X_{t-1} <X_t<1$.
      \item  If $X_{t-1} \geq 1$, then $x_i =0$ for all $i \in \prt_t$ and $X_t=X_{t-1} \geq 1$.
  \end{enumerate}
\end{lemma}
In other words, if period $t$ starts with cumulative effort $X_{t-1}<1$, the players exert strictly positive efforts, but the cumulative effort stays below $1$. On the other hand, if the cumulative effort is already $X_{t-1} \geq 1$, then all players choose zero effort and therefore $X_t = X_{t-1} \geq 1$. A straightforward implication of this lemma is that the total effort never reaches $1$ or above in equilibrium, and the individual efforts on the equilibrium path are always interior (i.e., strictly positive).
\begin{proof}
  If $X_{t-1} \geq 1$, then if any player $i$ in period $t$ chooses $x_i >0$, then $X_t > 1$ and therefore $X \geq X_t > 1$, which means that $u_i(\bx) < 0$. Since player $i$ can ensure zero payoff by choosing $x_i=0$, this is a contradiction. So, $x_i^*(X_{t-1})=0$ for all $X_{t-1} \geq 1$ and thus $X_t = X_{t-1} \geq 1$.

  Now, take $X_{t-1} < 1$. Suppose by contradiction that it leads to $X \geq 1$. This implies that in some period $s \geq t$ players chose efforts such that $X_{s-1} < 1$, but $X_s \geq 1$. This means that at least one player $i$ in period $s$ chose $x_i>0$ and gets a payoff of $u_i(\bx) \leq 0$. Now, there are two cases. First, if the induced total effort $X>1$, then player $i$'s payoff is strictly negative, and the player could deviate and choose $x_i=0$ to ensure zero payoff. On the other hand, if $X=1$, which means that $X_s=1$, then player $i$ could choose effort $\frac{x_i}{2}$, thus making $X_s<1$ and therefore $X < 1$, ensuring a strictly positive payoff. In both cases we arrive at a contradiction. Thus $X_{t-1} <1$ implies $X_t < 1$ and $X<1$.

  The last step is to show that $X_{t-1}<1$ implies $x_i>0$ for all $i \in \prt_t$. Suppose that this is not true, so that $x_i=0$ for some $i$. Then player $i$ gets a payoff of $0$. But by choosing $\widehat{x}_i \in (0,1-X_t)$, he can ensure that the cumulative effort $\widehat{X}_t = X_t + \widehat{x}_i < 1$ and thus the induced total effort $\widehat{X}<1$, and the new payoff of player $i$ is strictly positive. This is a contradiction.
\end{proof}

\begin{lemma} \label{L:characterization_necc}
  $X_t = f_t(X)$ is a necessary condition for equilibrium.
\end{lemma}

\begin{proof}
By \cref{L:characterization_01}, we only need to consider the histories with $X_{t-1}<1$. Moreover, we know that each player $i \in \prt_t$ chooses $x_i > 0$, i.e., an interior solution. Player $i$'s maximization problem is
\begin{equation*}
  \max_{x_i \geq 0}
  \frac{x_i}{f_t^{-1}(X_t)} - x_t
\end{equation*}
where $X_t = X_{t-1}+\sum_{j \in \prt_t} x_j$. Therefore, a necessary condition for optimum is
\begin{equation*}
  \frac{1}{f_t^{-1}(X_t)} - 1 + \frac{-x_i}{[f_t^{-1}(X_t)]^2} \dd{f_t^{-1}(X_t)}{X_t} = 0.
\end{equation*}
It is convenient to rewrite this condition in terms of the total effort $X$, taking into account that $X = f_t^{-1}(X_t)$ and $\dd{f_t^{-1}(X_t)}{X_t} = \frac{1}{f_t'(X)}$ to get
\begin{equation} \label{E:indFOC}
  x_i = f_t'(X) X(1-X)
  .
\end{equation}
Now, we can add up these necessary conditions for all players $i \in \prt_t$ and take into account that $f_t(X) = X_t = X_{t-1} + \sum_{i \in \prt_t} x_i$ to get a necessary condition for the equilibrium, $X_{t-1}=f_{t-1}(X)$, defined by \cref{E:ftrecdef}.
\end{proof}

\begin{lemma} \label{L:characterization_welldefined}
  Under \cref{C:cond1},
  $f_{t-1}^{-1}(X_{t-1})$ is well-defined, strictly increasing, and satisfies $f_{t-1}^{-1}(0)=\uX_{t-1}$ and $f_{t-1}^{-1}(1)=1$.
\end{lemma}

\begin{proof}
  First, note that even if $X_t$ would be $0$, the total effort induced by it would not be zero. In fact, by recursion it is straigtforward to show that it would be $\uX_t$. For any $X_{t-1}$ therefore $f_{t-1}^{-1}(X_{t-1}) \geq \uX_t$. Consequently, $X < \uX_t$ cannot be the total effort following any $X_{t-1}$.

  Moreover, by \cref{C:cond1}, $\uX_{t-1} \geq \uX_t$ and $f_{t-1}(X) <0$ for all $X \in [\uX_t,\uX_{t-1})$; therefore, total efforts in $[\uX_t,\uX_{t-1})$ range are not consistent with any $X_{t-1}$ either. We get that the only feasible range of the total effort $X$ induced by cumulative effort $X_{t-1}$ is $[\uX_{t-1},1]$. By \cref{C:cond1}, the function $f_{t-1}$ is continuously differentiable and strictly increasing in this range; therefore, the inverse is well-defined, continuously differentiable, and strictly increasing. Moreover, since $f_{t-1}(1)=1$, we have $f_{t-1}^{-1}(1)=1$, and since $\uX_{t-1}$ is a root of $f_{t-1}$, we have $f_{t-1}^{-1}(0)=\uX_{t-1}$.
\end{proof}

\begin{lemma} \label{L:characterization_brs}
  The best-response function of player $i \in \prt_t$ after cumulative effort $X_{t-1}$ is
  \begin{equation}
    x_i^*(X_{t-1}) =
    \begin{cases}
      \frac{1}{n_t} \left[ f_t(f_{t-1}^{-1}(X_{t-1})) - X_{t-1} \right]
        & \forall X_{t-1}<1,\\
      0
        & \forall X_{t-1}\geq 1.
    \end{cases}
  \end{equation}
  On the equilibrium path the individual efforts are $x_i^* = \frac{1}{n_t} \left[ f_t(X^*) - f_{t-1}(X^*) \right]$.
\end{lemma}

\begin{proof}
  \Cref{L:characterization_01} proved the claim for any $X_{t-1} \geq 1$. Take $X_{t-1}<1$. Then by \cref{L:characterization_01}, the individual efforts are interior, so they have to satisfy the individual first-order conditions \eqref{E:indFOC}.
  I showed that the total effort induced by $X_{t-1}$ must be $f_{t-1}^{-1}(X_{t-1})$. Inserting these results into the individual optimality condition for player $i \in \prt_t$ I get
  \begin{equation*}
  	x_i^*(X_{t-1})
  	= \frac{1}{n_t} \left[f_t(f_{t-1}^{-1}(X_{t-1}))-X_{t-1}\right].
  \end{equation*}
  In particular, on the equilibrium path, $X=X^*$, and therefore $x_i^* = \frac{1}{n_t} \left[f_t(X^*)-f_{t-1}(X^*)\right]$.
\end{proof}

So far, the arguments show that necessary conditions for equilibria lead to a unique candidate for equilibrium---the strategies specified in the theorem. Finally, we have to check that this is indeed an equilibrium. That is, we need to show that all players are indeed maximizing their payoffs.

\begin{lemma} \label{L:characterization_SOC}
$\bx^*$ is an equilibrium.
\end{lemma}
\begin{proof}
By construction $x_i^*(X_{t-1})$ is a local extremum for player $i \in \prt_t$, given that the cumulative effort prior to period t is $X_{t-1}$ and all other players behave according to their equilibrium strategies.
Since the local extremum is unique and ensures strictly positive payoff (which is strictly more than zero from corner solution $x_i=0$), $x_i^*(X_{t-1})$ is also the global maximum. Thus, no player has an incentive to deviate.
\end{proof}

\subsection{Proof of the Information Theorem (\Cref{T:information})}

For any contest $\bn=(n_1,\dots,n_T)$ and any period $t \leq T$, let $\bn^t = (n_{t+1},\dots,n_T)$ denote the sub-contest starting after period $t$. In particular, $\bn^T = \varnothing$ and $\bn^0=\bn$.
Note that  $f_t(X)$ depends only on $\bn^t$.

Remember that $g_1,\dots,g_T$ are recursively defined as $g_1(X) = X(1-X)$ and $g_{k+1}(X)=-g_k'(X) X(1-X)$, so they are independent of $\bn$. Also, $\bS(\bn) = (S_1(\bn),\dots,S_T(\bn))$ are defined so that $S_k(\bn)$ is the sum of all products of $k$-combinations of vector $\bn$ and is therefore independent of $X$. Similarly, $\bS(\bn^t)$ is defined in the same way for each subcontest.

The proof has two key steps. The first step (\cref{L:ftmeasures}) shows that we can express the inverted best-response functions through a weighted sums of measures of information as in \cref{E:ft_withSg}.

In particular, a sufficient condition that guarantees $X^*$ is strictly increasing in $\bS(\bn)$ is that each $g_k(X^*)>0$, i.e., the efforts are \emph{higher-order strategic substitutes} near equilibrium. The following \cref{C:cond2} formalizes this idea with a small caveat. Namely, as we will see below, in the case of fully the sequential contest $g_n(X^*)=0$ and therefore, the strict version of this definition is not satisfied. However, as I will prove, for any $n > 2$, it is sufficient that efforts are weak strategic substitutes of level $n$ as defined below.

\begin{assumption}[$K$-th order strategic substitutes near equilibrium] \label{C:cond2}
	Efforts are (weak) strategic substitutes of level $K$ near equilibrium, if $g_K(X^*) \geq 0$ and $g_k(X^*) > 0$ for all $k \in \{2,\dots,K-1\}$ at the equilibrium level of total effort $X^*$. Efforts are strict strategic substitutes of level $K$ near equilibrium, if they are strategic substitutes of level $K$ and $g_K(X^*) > 0$. 
\end{assumption}

\begin{lemma} \label{L:ftmeasures}
  The function $f_t(X)$ can be expressed as 
  \begin{equation} \tag{\ref{E:ft_withSg}}
    f_t(X) = X - \sum_{k=1}^{T-t} S_k(\bn^t) g_k(X).
  \end{equation}
\end{lemma}
\begin{proof}
  By construction, the subcontest starting after period $T$ has no players, so $S_k(\bn^T) = 0$ for any $k$. Therefore $f_T(X) = X$ satisfies \cref{E:ft_withSg}. Now, suppose that the characterization holds for $f_t(X)$. Then, since $g_{k+1}(X) = -g_k'(X) X(1-X)$, we get that
  \begin{align*}
  	f_t'(X) X(1-X)
  	&= X(1-X) - X(1-X) \sum_{k=1}^{T-t} S_k(\bn^t) g_k'(X)\\
  	&= g_1(X) + \sum_{k=2}^{T-t+1} S_{k-1}(\bn^t) g_k(X)
  	.
  \end{align*}
  Therefore, $f_{t-1}(X) = f_t(X) - n_t f_t'(X) g(X)$ implies that
  \begin{align*}
  	f_{t-1}(X)
  	&= X - \sum_{k=1}^{T-t} S_k(\bn^t) g_k(X)
  	- n_t g_1(X) - n_t \sum_{k=2}^{T-t+1} S_{k-1}(\bn^t) g_k(X)\\
  	&= X - [S_1(\bn^t)+n_t] g_1(X)
  	- \sum_{k=2}^{T-t} [S_k(\bn^t)+n_t S_{k-1}(\bn^t)] g_k(X)
  	- n_t S_{T-t}(\bn^t) g_{T+1}(X)
  	.
  \end{align*}
  Note that $S_1(\bn^t) = \sum_{s>t} n_s$ and $g_1(X)=g(X)$, so that $S_1(\bn^t) + n_t =S_1(\bn^{t-1})$. Similarly, $\bn^{t-1} = (n_t,\bn^t)$, so $S_k(\bn^t)$ includes all $k$-combinations of $\bn^{t-1}$ except the ones involving $n_t$. Adding $n_t S_{k-1}(\bn^t)$ therefore completes the sum, so that $S_k(\bn^{t-1}) = S_k(\bn^t) + n_t S_{k-1}(\bn^t)$. Since $S_{T-t}(\bn^t) = n_{t+1} \dots n_T$, we have that $n_t S_{T-t}(\bn^t) =
  n_t \times \dots \times n_T = S_{T-(t-1)}(\bn^{t-1})$. Therefore, we can express $f_{t-1}(X)$ as
  \begin{equation*}
  	f_{t-1}(X)
  	= X - \sum_{k=1}^{T-(t-1)} S_k(\bn^{t-1}) g_k(X).
  \end{equation*}
\end{proof}

\begin{proposition}[Efforts are higher-order strategic substitutes] \label{P:cond2_tullock}
Take a contest $\bn$ with $T \leq n$ periods with positive number of players. Then
\begin{enumerate}
	\item If $T<n$, efforts are strict strategic substitutes of level $T$ near equilibrium.
	\item If $T=n$, efforts are weak strategic substitutes of level $n$ near equilibrium.
\end{enumerate}
\end{proposition}

As the proof is long, I include the proof as a separate subsection below.
With these results, the proof of the information theorem is now straightforward.

\begin{proof}[Proof of \cref{T:information}]
\newcommand{\hT}{\widehat{T}}
Take two contests $\bn$ and $\hbn$ such that $\bS(\hbn) > \bS(\bn)$, i.e., $S_k(\hbn) \geq S_k(\bn)$ for all $k$ and the inequality is strict for at least one $k$. Let $T$ and $\hT$ be the number of periods with strictly positive number of players in contests $\bn$ and $\hbn$ respectively. Notice that by assumptions, $T \leq \hT$ and $S_k(\bn) = 0$ for all $k > T$.		
By \cref{T:characterization}, the total equilibrium $X^*$ is the highest root of $f_0(X)$ in $[0,1]$. By \cref{L:ftmeasures}, we can express $f_0(X^*)$ as
\begin{equation*}
	f_0(X^*) = X^* - \sum_{k=1}^{T} S_k(\bn) g_k(X^*)
	= X^* - \sum_{k=1}^{\hT} S_k(\bn) g_k(X^*).
\end{equation*}
Similarly, let $\hX^*$ be the total equilibrium effort in contest $\hbn$. It is the highest root of $\hf_0(X^*)$ in $[0,1]$, which we can write as
\[
\hf_0(\hX^*)  = \hX^* - \sum_{k=1}^{\hT} S_k(\hbn) g_k(\hX^*)
.
\]
Suppose by contradiction that the claim of the theorem does not hold and so $\hX^* \leq X^*$. Since by \cref{C:cond1} $\hf_0$ is strictly increasing in $[\hX^*,1]$, we get that $0 = \hf_0(\hX^*) \leq \hf_0(X^*)$. Therefore
\[
0
\leq 
\hf_0(X^*)
- 
f_0(X^*)
=
- \sum_{k=1}^{\hT} [S_k(\bn)-S_k(\bn)] g_k(X^*)]
\]
As $S_k(\bn) \geq S_k(\bn)$ and $g_k(X^*)\geq 0$ for each $k \in \{1,\dots,K\}$, the sum on the right-hand side is non-positive. 
Moreover, of at least one $k$, we have $S_k(\bn)>S_k(\bn)$. Now, by \cref{P:cond2_tullock}, if $\hT < n$, the efforts are strict strategic substitutes, so $g_{k}(X^*)>0$, therefore we get a contradiction. 

Finally, suppose that $\hT=n$. As the only $n$-player contest with a positive number of players in $n$ periods is the fully sequential contest, we must have $\hbn = (1,1,\dots,1)$.
It is straightforward to verify that then $S_n(\hbn)=1$ and $S_{n-1}(\hbn)=n$. Now, notice that since the contest $\bn$ is strictly less informative than $\hbn$, it must have at least one period with two players. Let us replace it with a new contest $\bn'$, where we have split all players into separate periods and left only one period with two players, i.e., the contest $\bn'$ is a permutation of $(2,1,1,\dots,1)$. Clearly, the contest $\bS(\bn') \geq \bS(\bn)$. As in both contests $\bn$ and $\bn'$ the number of periods is strictly less than $n$, the part of the theorem we already proved implies that the corresponding equilibrium effort $X^{*'} \geq X^*$.

In contest $\bn'$, $S_n(\bn') = 0 < S_n(\hbn)$, $S_{n-1}(\bn') = 2 < n S_{n-1}$, and $S_k(\bn') \leq S_k(\hbn)$ for all $k<n-1$. As by \cref{P:cond2_tullock}, the efforts are weak strategic substitutes of level $n$ near equilibrium, this proves that $X^* \leq X^{*'} < \hX^* \leq X^*$. This is a contradiction.
\end{proof}

Remark: the last paragraph of the proof shows why we need the assumption that $n> 2$. Otherwise, in two-player contest, the sequential contest implies $\bS(1,1) = (2,1)$ and simultaneous contest $\bS(2) = (2,0)$. These two contests only differ by the measure of information of level $2$. As the efforts are only weakly strategic substitutes at $X^*$, the proof would not be valid. Indeed, it is straightforward to check that with $n=2$, $X^*=0.5$ and $g_2(0.5) = 0$, so the two contest would give the same total effort. With $n=3$, this issue does not arise, as $S_2(1,1,1) = 3$ and with any other three-player contest $S_2(\bn) \leq 2$.

\subsection{Proof That Efforts Are Higher-Order Strategic Substitutes (\Cref{P:cond2_tullock})}

Remember that $g_1(X)=X(1-X)$ and $g_k(X) = -g_{k-1}'(X) X (1-X)$ for all $k > 1$. Therefore, $g_1(X)$ is a second-degree polynomial, $g_2(X)$ third-degree, and so on. In particular, $g_k(X)$ is a polynomial of degree $k+1$ and therefore has up to $k+1$ real roots. In the following, I show that all roots are real and in $[0,1]$. Let these roots be denoted as $Z_{0:k} \leq Z_{1:k} \leq \dots \leq Z_{k:k}$. The proof keeps track of the order and locations of these roots and their comparison with $X^*$.

\begin{proof}[Proof of \cref{P:cond2_tullock}]
    The proof relies on three lemmas that I prove below. 
    
    \begin{enumerate}
        \item \Cref{L:tullock_g_props} shows that the highest root of $g_k$ is $Z_{k:k}=1$, the second highest $Z_{k-1:k} \in (Z_{k-2:k-1},1)$, and $g_k(X)>0$ for all $X$ between the highest two roots.
        Therefore, to prove that $g_k(X^*)>0$, it suffices to show that $X^*>Z_{k-1:k}$.
        \item \Cref{L:ft_gt_tullock} establishes a connection between the total equilibrium effort $X^*$ and $Z_{k-1:k}$. It shows that if we take the sequential $n$-player contest $\bn=(1,\dots,1)$, then $f_{n-k}(X) = g_k(X) \frac{X}{1-X}$ for all $k =1,\dots,n$. Therefore, if we take the fully sequential contest with $n$ players, we get $f_0(X)=g_n(X) \frac{X}{1-X}$, and so the total equilibrium effort $X^*$ of this contest is exactly equal to the second highest root of $g_n$, i.e., $Z_{n-1:n}$.

        This proves the ``weak'' part of the proposition, i.e., if $\bn$ is fully sequential, then $X^* = Z_{n-1:n}$, which is a root of $g_n$, and therefore $g_k(X^*)=0$.
        \item \Cref{L:Xmonotone} shows directly\footnote{Note the first part of \cref{C:information} proves the same claim, but since \cref{P:cond2_tullock} establishes a sufficient condition for \cref{T:information} and hence its \cref{C:information}, to avoid a circular argument I prove it here directly.} that $X^*$ is strictly increasing in each $n_t$. Therefore, if the contest is not sequential ($n_t>1$ for some $t$), then the total effort in this contest is strictly higher than in the fully sequential $T$-player contest. Thus, $X^* > Z_{T-1:T}$ and $g_T(X^*) > 0$.
        \item Finally, \cref{L:tullock_g_props} also shows that the adjacent $g_k$'s are interlaced; i.e., the second highest roots are increasing in $k$, so that for all $k<T$, $Z_{k-1:k} < Z_{T-1:T} \leq X^*$, and therefore $g_k(X^*) > 0$ for all $k<T$.
    \end{enumerate}
\end{proof}

\begin{lemma} \label{L:tullock_g_props}
    Each $g_k$ has the following properties:
    \begin{enumerate}
        \item $g_k(1) = g_k'(1) = -1$.
        \item $g_k$ can be expressed as
        \begin{equation}
        g_k(X) = -\prod_{j=0}^k (X-Z_{j:k}),
        \end{equation}
        where $0=Z_{0:k} < Z_{1:k} < \dots < Z_{k:k} = 1$.
        \item $Z_{s:k+1} \in (Z_{s-1:k},Z_{s:k})$ for all $s=1,\dots,k$.
    \end{enumerate}
\end{lemma}

\begin{proof}
    First note that $g_1(X)=g(X)=X(1-X)$ is a polynomial of degree $2$. Each step of the recursion gives a polynomial of one degree higher; i.e., $g_k(X)$ is a polynomial of degree $k+1$, so $g_k'(X)$ is a polynomial of degree $k$, and therefore $g_{k+1}(X) = -g_k'(X) X (1-X)$ is a polynomial of degree $k+2$.
    \begin{enumerate}
        \item $g_{k+1}(1) = -g_k'(1) g(1) = 0$, because $g(1)=1 (1-1)=0$. Therefore
        $g_k'(1)
        = -g_{k-1}''(1) g(1) - g'(1) g_{k-1}'(1)
        = g_{k-1}'(1)
        = \dots
        = g_1'(1)
        = g'(1)
        = 1-2 \cdot 1=-1$.
        \item The claim clearly holds for $g_1(X) = X(1-X)$ with $Z_{0:1}=0<Z_{1:1}=1$. Suppose it holds for $k$.
        Since all $k+1$ roots of $g_k$ are real and in $[0,1]$, by the Gauss-Lucas theorem all $k$ roots of $g_k'$ are in $(0,1)$.
        Then $g_{k+1}(X) = -g_k'(X) X (1-X)$ clearly has roots at $0$ and $1$ and $k$ roots in $(0,1)$. To see that the roots are all distinct, note that
        \begin{equation*}
        g_k'(X) = - \sum_{s=0}^k \prod_{j \neq s} (X-Z_{j:k}).
        \end{equation*}
        Therefore,
        $g_k'(Z_{s:k}) = - \prod_{j \neq s} (Z_{s:k}-Z_{j:k})$,
        which is strictly negative for $s=k$, strictly positive for $s=k-1$, and so on. Therefore, for each $s=1,\dots,k$, function $g_k'$; hence,  $g_{k+1}$ also has a root $Z_{s:k+1} = (Z_{s-1:k},Z_{s:k})$. This determines the $k$ interior roots.
        \item The previous argument also proves the last claim.
    \end{enumerate}
\end{proof}

\begin{lemma} \label{L:ft_gt_tullock}
    If $\bn = (1,\dots,1)$, then $f_{n-k}(X) = g_k(X) \frac{X}{1-X}$ for all $k=1,\dots,T$.
\end{lemma}
\begin{proof}
    Suppose that $\bn = (1,\dots,1)$.
    First, $f_{n-1}(X) = X - X (1-X) = X^2 = g_1(X) \frac{X}{1-X}$. Now, suppose that $f_{n-k}(X) = g_k(X) \frac{X}{1-X}$. Then since
    \begin{equation*}
    \dd{\frac{X}{1-X}}{X} X (1-X)
    = \left[ \frac{1}{1-X} - \frac{-X}{(1-X)^2}\right] X (1-X)
    = \frac{X}{1-X},
    \end{equation*}
    we get that
    \begin{align*}
    f_{n-(k+1)}(X)
    &= g_k(X) \frac{X}{1-X}
    - g_k(X) \dd{\frac{X}{1-X}}{X} X(1-X)
    - g_k'(X) \frac{X}{1-X} X(1-X)\\
    &= g_{k+1}(X) \frac{X}{1-X}.
    \end{align*}
\end{proof}

\begin{lemma} \label{L:Xmonotone}
    $X^*$ is strictly increasing in each $n_t$.
\end{lemma}
\begin{proof}
    I first show that $X^*$ is independent of permutations of $\bn$. Fix a contest $\bn$ and a period $t>1$.
    To shorten the notation, let $\phi_t(X) = f_t'(X) X (1-X)$.
    \begin{align*}
    f_{t-1}(X)
    &
    = f_t(X) - n_t \phi_t(X), \\
    f_{t-1}'(X)
    &= f_t'(X) - n_t \phi_t'(X) 
    = \frac{\phi_t(X)}{g(X)} - n_t \phi_t'(X), 
    \\
    f_{t-2}(X)
    &
    = f_t(X) - [n_{t-1}+n_t] \phi_t(X) + n_{t-1} n_t \phi_t'(X) X (1-X).
    \end{align*}
    Switching $n_{t-1}$ and $n_t$ in $\bn$ does not affect $f_{t-2}$, and therefore it also doesn't affect  $f_0$. This means that any such switch leaves $X^*$ unaffected, which means that $X^*$ is independent of permutations of $\bn$.

    To prove that $X^*$ is strictly increasing in each $n_t$, it therefore suffices to prove that it is strictly increasing on $n_1$. Take $\hbn = (n_1+1,n_2,\dots,n_T)$. Then  $f_1$ is unchanged and the corresponding $\hf_0$ at the original equilibrium $X^*$ is
    \begin{equation*}
    \hf_0(X^*)
    = f_1(X^*) - (n_1+1) f_1'(X^*) X^* (1-X^*)
    = f_0(X^*) - f_1'(X^*) X^* (1-X^*)
    < 0,
    \end{equation*}
    because $f_0(X^*)=0$ and $f_1(X^*) > 0$ by \cref{C:cond1}. By \cref{C:cond1}, $\hf_0$ is strictly increasing between its highest root $\hX^*$ and $1$, thus $\hX^* > X^*$.
\end{proof}

\subsection{Proof of Decreasing Weights Lemma (\Cref{L:gktullock})}

This lemma allows to order some contests, which cannot be ranked according to their information measures. For example, two 10-player contests $\bn = (5,5)$ and $\hbn = (8,1,1)$ have corresponding information measures $\bS(\bn)=(10,25)$ and $\bS(\hbn)=(10,17,8)$. Contest $\bn$ has more second-order information, but $\hbn$ has one more disclosure and thus more third-order information. However, the sum of all information measures is $10+25 = 10+17+8=35$. Since the weights are higher in lower-order information, this implies that the total effort is higher in the first contest. Indeed, direct application
\cref{T:characterization} confirms this, as $X^*= \frac{13+\sqrt{41}}{20} \approx 0.9702 > \hX^* = \frac{31+\sqrt{241}}{48} \approx 0.9693$.

\begin{proof}[Proof of \cref{L:gktullock}]
	By \cref{L:ft_gt_tullock}, 
	$g_k(X) = \hf_{\hn-k}(X) \frac{1-X}{X}$, where $\hf_{\hn-k}$ is defined for a sequential $\hn \geq k$-player contest. Similarly, $g_{k-1}(X) = \hf_{n+1-k}(X) \frac{1-X}{X}$. Therefore,
	\[
	g_{k-1}(X^*) - g_k(X^*)
	= [\hf_{\hn+1-k}(X^*)-\hf_{\hn-k}(X^*)] \frac{1-X^*}{X^*}
	= \hf_{\hn+1-k}'(X^*) (1-X^*)^2
	.
	\]
	Now, take $\hn = T$. Then by \cref{L:Xmonotone}, $X^*$ is weakly higher than the highest root of $\hf_0$. By \cref{C:cond1}, the highest root of $\hf_{T+1-k}$ is even (weakly) lower and $\hf_{T+1-k}$ is strictly increasing above its highest root, so that $\hf_{\hn+1-k}'(X^*) > 0$. This proves that $g_{k-1}(X^*) > g_k(X^*)$.  
\end{proof}

\subsection{Proofs of Implications of the Information Theorem (\Cref{C:information})}

\begin{proof}[Proof of \cref{C:information}]
	Take two contests $\bn^1$ and $\bn^2$ and let $X^1$ and $X^2$ be the corresponding total equilibrium efforts.
	\begin{enumerate}
		\item Suppose that $\bn^1 < \bn^2$. Then $\bS(\bn^1) < \bS(\bn^2)$ and therefore $X^1 < X^2$.
		\item If $\bn^1$ is a permutation of $\bn^2$, then $\bS(\bn^1) = \bS(\bn^2)$ and therefore $X^1=X^2$.
		\item If $\prts^1$ is a coarser partition than $\prts^2$, then $\bS(\bn^1) < \bS(\bn^2)$ and therefore $X^1 < X^2$.
		\item If $\sum_t n_t^1 = \sum_t n_t^2 = n$ and there exist $t,t'$ such that $n_t^1 n_{t'}^1 < n_t^2 n_{t'}^2$ and $n_s^1 = n_s^2$ for all $s \neq t,t'$, then by construction $S_1(\bn^1)=S_2(\bn^2) = n$ and $S_k(\bn^1) < S_k(\bn^2)$ for all $k>1$. Therefore $X^1 < X^2$.
		\item Let $\bn^1 = (n)$. Then for any $\bn^2 \neq \bn^1$, $\bS(\bn^1) < \bS(\bn^2)$, so indeed $X^1$ is the unique minimum of $X^*$ over all contests. Similarly, if $\bn^2 = (1,1,\dots,1)$, any other contest has strictly lower measures of information and therefore $X^2$ is the unique maximum of $X^*$ over all contests.
		
		To establish the final claim of the optimality of equal division of players, let $\bn^1$ be $n$-player contests where players are distributed among at most $T$ periods. Suppose by contradiction that the corresponding total equilibrium effort $X^*$ is a maximum over all such contests and $\bn^1$ does not split players as equally as possible. In particular, let $k = \lfloor n/T \rfloor$. Equal split requires that each period has either $n_t^1 \in \{k,k+1\}$ players. Since this is not the case, there exists a period $t$ where $n_t^1 \leq k-1$ and a period $s$ where $n_s^1 \geq k+1$ (or $t,s$ such that $n_t^1 \leq k$ and $n_s^1 \geq k+2$, then the proof is analogous). 
		
		We can now construct a new contest, $\bn^2$, where we have moved one player from period $s$ to period $t$. Then
		as $n_s^1-1 \geq k > n_t^1$,  
		\[
		n_t^2 n_s^2
		= 
		(n_t^1+1) (n_s^1-1)
		=
		n_t^1 n_s^1
		-n_t^1
		+n_s^1
		-1
		>
		n_t^1 n_s^1.
		\]
		Therefore the contest $\bn^2$ is more homogeneous than $\bn^1$ and so $X^1 < X^2$ by the previous step. Thus, we found a contradiction with the assumption that $X^1$ is a maximal total effort among such contests.
	\end{enumerate} 
\end{proof}

\subsection{Proof of the Earlier-Mover Advantage (\Cref{P:earliermover})}

\begin{proof}[Proof of \cref{P:earliermover}]
	The equilibrium payoff of player $i$ is $u_i(\bx^*) = x_i^* \left( \frac{1}{X^*}-1 \right)$, so the payoffs are ranked in the same order as the individual efforts (in fact they are proportional to individual efforts).
	Therefore, it suffices to prove that if $i \in \prt_t$ and $j \in \prt_{t+1}$, then $x_i^* > x_j^*$.
	Using \cref{T:characterization} and \cref{E:ft_withSg}, the difference in equilibrium efforts can be expressed as
	\begin{equation*}
		x_i^* - x_j^*
		= \sum_{k=1}^{T-t} \left[ S_k(\bn^t) - S_k(\bn^{t+1}) \right] g_{k+1}(X^*).
	\end{equation*}
	Now, note that $\bS(\bn^t) \geq \bS(\bn^{t+1})$ as there is less information remaining in the game that starts one period later. Moreover, $S_1(\bn^t) > S_1(\bn^{t+1})$ as $\bn^t$ includes player $j$, whereas $\bn^{t+1}$ does not. Finally note that by \cref{P:cond2_tullock}, $g_2(X^*) > 0$ and therefore $x_i^* - x_j^* > 0$.
\end{proof}

\subsection{Proof of the Large Contests Limit (\Cref{P:largecontests})}

\begin{proof}[Proof of \cref{P:largecontests}]
	By \cref{T:characterization}, each $X^n < 1$.
	On the other hand, by \cref{T:information}, $X^n \geq \frac{n-1}{n}$, which is the total equilibrium effort of the simultaneous $n$-player contest (see \cref{S:example}). Therefore $\lim_{n \to \infty} X^n = 1$.
	
	The total equilibrium effort of a censored contest $\bn^n$ is the highest root of $f_0(X)$, which can be expressed by \cref{E:f0_withSg} as 
	\begin{equation} \label{E:f0_withSg2}
		X^n = \sum_{k=1}^T S_k(\bn^n) g_k(X^n).
	\end{equation}	
	For each $k$, function $g_k(X)$ is a twice continuously differentiable function (a polynomial), $g_k(1) = 0$, and
	$g_k'(1) = - g_{k-1}''(1) g(1) - g_{k-1}'(1) g'(1) = g_{k-1}'(1) = \dots = g_1'(1) = - 1$, as $g_1(X)=X(1-X)$.
	Therefore, for all $k > 1$,
	\begin{equation*}
		\lim_{X \to 1} \frac{g_k(X)}{X(1-X)}
		=
		\lim_{X \to 1} \frac{-g_{k-1}'(X) X(1-X)}{X(1-X)}
		=
		-g_{k-1}'(1)
		=
		1.
	\end{equation*}
	Taking limits from both sides of \cref{E:f0_withSg2} and using the result that $\lim_{n \to \infty} X^n = 1$, 
	\begin{align*}
		1 &= \lim_{n \to \infty} X^n
		= \lim_{n \to \infty} \sum_{k=1}^T S_k(\bn^n) \frac{g_k(X^n)}{X^n (1-X^n)} X^n (1-X^n)
		= \lim_{n \to \infty} (1-X^n) \sum_{k=1}^T S_k(\bn^n) .
	\end{align*}
	To shorten the notation, let $S^n = \sum_{k=1}^T S_k(\bn^n)$. Rearranging the previous equation gives
	\begin{align} \label{E:Slimit}
		0 &= \lim_{n \to \infty} [1-(1-X^n)S^n ]
		=	\lim_{n \to \infty}  \left[ X^n - \left(1-\frac{1}{S^n} \right) \right] S^n.
	\end{align}
	We can express $S^n = \sum_{k=1}^T S_k(\bn^n) = \prod_{t=1}^T (1+n_t^k) - 1$. As $\lim_{n \to \infty} S^n = \infty$, \cref{E:Slimit} implies that
	\[
	\lim_{n\to \infty} \left[ X^n - \left(1-\frac{1}{S^n} \right) \right] 
	=
	\lim_{n\to \infty} \left[ X^n - \left(1-\frac{1}{\prod_{t=1}^T (1+n_t^n)} \right) \right] 
	= 0.
	\]
	For individual effort of player $i \in \prt_t^n$, we can use \cref{T:characterization} and \cref{E:ft_withSg} to get
	\begin{align*}
		x_i^n
		=  
		g_1(X^n) + \sum_{k=1}^{T-t} S_k(\bn^t) g_{k+1}(X^n).
	\end{align*}
	Taking the limit, again using the facts that $X^n \to 1$ and $ \frac{g_{k+1}(X^n)}{X^n (1-X^n)} \to 1$,
	\begin{align*}
		\lim_{n \to \infty} x_i^n
		&
		=
		\lim_{n \to \infty} (1-X^n)	\left[
		1 + \sum_{k=1}^{T-t} S_k(\bn^t) 
		\right].
	\end{align*}
	Now, note that $1 + \sum_{k=1}^{T-t} S_k(\bn^t) = \prod_{s=t}^T (1+n_s^n)$. Therefore, using the result from above, we can express the last equation as
	\begin{align*}
		0
		&=
		\lim_{n \to \infty} \left[ 
		x_i^n
		-
		(1-X^n) \left(
		1 + \sum_{k=1}^{T-t} S_k(\bn^t) 
		\right)
		\right] 
		= \lim_{n \to \infty} \left[ 
		x_i^n
		-
		\frac{1}{\prod_{s=t}^T (1+n_s^n)}
		\right] 
	\end{align*}
\end{proof}

\clearpage

\section{General Results} \label{A:general}

In this appendix, I show that the approach introduced in the paper applies more generally.
As in the main text, the set of players $\prt = \{1,\dots,n\}$ is partitioned into groups $\{\prt_1,\dots,\prt_T\}$, such that all $n_t = \# \prt_t$ players in group $t$ observe the cumulative sum of actions $X_{t-1}=\sum_{s=1}^{t-1} \sum_{j \in \prt_s} x_j$ of players in all previous groups. The total action $X = \sum_{i=1}^n x_i$. Before any player has chosen the action, the cumulative action is $X_0 = 0$ and after all players have chosen their actions, the cumulative action is $X_T=X$.

Suppose that each player chooses an action $x_i$ from a set $\gX_i$ and if the profile of actions is $\bx = (x_1,\dots,x_n)$, then player $i$ gets a payoff
\begin{equation}
	U_i(\bx)
	= u_i(x_i,X).
\end{equation}
Take a player $i$ from the last period $T$. Player $i$ observes cumulative effort $X_{T-1}$ before period $T$ and knows that other players in period $T$ are choosing efforts simultaneously to him. Therefore he maximizes
\[
\max_{x_i \in \gX_i} u_i\left(x_i,x_i + X_{T-1} + \sum_{j \in \prt_T \setminus \{i\}} x_j \right).
\]
The standard best-response function would be $x_i^*(X_{T-1})$.\footnote{More formally, this function is called reduced best-response function as it only depends on the sum.} But suppose we can express the optimal effort $x_i$ choice as a function of total effort, $\phi_i(X)$. Then adding up individual efforts in period $T$ consistent with total effort $X$ gives us a necessary condition for equilibrium,
\[
X_{T-1}
=
X - \sum_{i\in \prt_T} \phi_i(X).
\]
I denote the function on the right-hand-side by $f_{T-1}(X)$. Its inverse function (assuming it exists), $f_{T-1}^{-1}(X_{T-1})$ is the total effort induced by cumulative effort $X_{T-1}$, if all players in period $T$ behave optimally.

Suppose by induction that the same argument holds starting from period $t$, i.e., if cumulative effort after $t$ is $X_t$ then the total effort induced is $f_t^{-1}(X_t)$. Then player $i$ in period $t$ maximizes
\[
\max_{x_i \in \gX_i} u_i\left( x_i,f_t^{-1}(X_t) \right).
\]
If again, we can express the optimal $x_i$ only as a function $\phi_i(X)$, then adding up the conditions would give us a necessary condition for equilibrium
\[
X_{t-1}
= X_t - \sum_{i \in \prt_t} x_i
= f_t(X) - \sum_{i \in \prt_t} \phi_i(X),
\]
which I denote by $f_{t-1}(X)$. Finally, in the beginning of the game cumulative total action is $X_0=0$, which gives us an equilibrium condition for the whole game.

There are some gaps in this analysis that need to be filled.
In the paper I showed with the Tullock contest payoffs, $u_i(x_i,X)=\frac{x_i}{X}-x_i$ and $\gX_i = \R_+$, all the necessary assumptions are satisfied, so that this analysis characterizes the unique equilibrium.
In the next subsections, I show how the analysis can be applied in several extensions of the model in the paper.
I also give some examples where the analysis fails.

\subsection{Aggregative Games} \label{A:aggregative}

In a special case when all players make their choices simultaneously, i.e., when $T=1$, the game defined here is a \emph{(linearly) aggregative game}, introduced by \cite{selten_preispolitik_1970}.\footnote{For a detailed review, see \cite{jensen_aggregative_2018}.} In this case, my construction simply requires that
\begin{equation} \label{E:aggregate_equilibrium}
	f_0(X)
	= X - \sum_{i=1}^n \phi_i(X)
	= 0
	\iff
	X = \sum_{i=1}^n \phi_i(X).
\end{equation}
This is a known condition, in this case $\phi_i(X)$ is called backward response correspondence of player $i$ and $\sum_{i=1}^n \phi_i(X)$ the aggregate backward correspondence. If the utility functions $u_i$ are quasi-concave and upper semi-continuous in $x_i$, then the game has an equilibrium \citep{jensen_aggregative_2018}.

However, if $T>1$, then the game is not aggregative anymore because player $i$'s payoff is affected differently by different players, depending on the period they are moving. In this sense, the analysis here is a dynamic generalization of linearly aggregative games.

\subsection{Linearly Multiplicative Payoffs} \label{A:xhX}

The results generalize directly to a class of payoff functions, where the utility is linearly multiplicative in players' own action,
\begin{equation}
	u_i(x_i,X) = x_i h(X) \label{E:xhX}, \;\;\; x_i \in \gX_i = \R_+.
\end{equation}
In the case of Tullock contest payoffs, $h(X)=\frac{v}{X}-c$, where $v$ is the value of the prize and $c$ is the marginal cost of effort. Another important application for this class of games is an oligopoly with linear cost, where $h(X) = P(X)-c$, where $x_i$ is the firm's own quantity, $X$ the total quantity, $P(X)$ the inverse demand function, and $c$ the marginal cost.
Finally, it also includes public goods games, where $x_i$ is the private consumption and $h(X)$ is the marginal benefit of private consumption, which is decreasing in the public good contributions and therefore decreasing in total private consumption.

It is natural to assume in these applications that $h(X)$ is strictly decreasing up to some upper bound $\oX$, at which it takes value $h(\oX)=0$ and above which $h(X) \leq 0$. Therefore, effectively the action space is $\gX_i = [0,\oX]$. Without loss of generality, we can change the scale of actions so that $\oX = 1$.

The first-order optimality condition for players in period $T$ is then
\[
h(X) + x_i h'(X) = 0
\iff
x_i = g_1(X),
\]
where $g_1(X)=-\frac{h(X)}{h'(X)}$. Therefore we can write the inverted best-response function as
\[
f_{T-1}(X) = X-n_T g_1(X).
\]
Similarly, if the inverted best-response functions at period $t$ is $f_t(X)$, which is invertible in the relevant range, the payoff function of player $i$ in period $t$ is $u_i(x_i,f_t^{-1}(X_t))=x_i h(f_t^{-1}(X_t))$ and therefore the first-order condition for players in period $t$ is
\begin{equation} \label{E:FOClinmult}
	h(X) + x_i h'(X) \frac{1}{f_t'(X)} = 0
	\iff
	x_i = g_1(X) f_t'(X).
\end{equation}
Therefore $f_{t-1}(X)=f_t(X)- n_t f_t'(X) g_1(X)$. This shows that we can use the characterization derived in the paper, with two modifications. First, instead of specific expression $X(1-X)$, we have a function $g_1(X) = -\frac{h(X)}{h'(X)}$. And second, we need to impose some conditions on the function $h(X)$ so that the conditions for the existence and uniqueness are satisfied.

In particular, if \cref{C:cond1} holds, i.e., $f_t$ functions are well-behaved, then the characterization theorem (\cref{T:characterization}) holds without any modifications. Therefore, the equilibrium is still unique and can be computed as the highest root of $f_0(X)$ in $[0,1]$. Moreover, the limit for large contests (\cref{P:largecontests}) holds as well, with a particular adjustment in formulas. Let $\alpha = -g_1'(1) > 0$. Then the formulas in \cref{E:largecontests} would be adjusted as
\begin{equation}
	\lim_{n \to \infty} \left[
	X^n - \left( 1-\frac{1}{\prod_{t=1}^T (1+\alpha n_t^n)} \right)
	\right]
	= 0
	\;\;
	\text{and}
	\;\;
	\lim_{n \to \infty} \left[ x_i^n - \frac{\alpha}{\prod_{s=1}^t (1+ \alpha n_s^n)} \right] = 0.
\end{equation}

The adjustment of $g_k$ functions is also straightforward, defined as $g_1(X) = -\frac{h(X)}{h'(X)}$ and $g_{k+1}(X) = -g_k'(X) g_1(X)$ for all $k$. If in addition to \cref{C:cond1}, also \cref{C:cond2} holds, i.e., efforts are higher-order strategic substitutes, then essentially all the remaining results in the paper generalize directly.
In particular,  the information theorem (\cref{T:information}), all its corollaries (\cref{C:information}), and the earlier mover advantage result (\cref{P:earliermover}) hold as stated.

The only result that does not generalize is \cref{L:gktullock} that showed that with Tullock payoffs, the weights $g_k(X^*)$ are decreasing in $k$. It is easy to see that this result depends on the function $h(X)$. 
For example, consider the case when $h(X) = \sqrt[\alpha]{1-X}$ for all $X \in [0,1]$ and $0$ otherwise, where $\alpha > 0$ is a constant. 
Then $g_1(X) = \alpha(1-X)$, $g_2(X)=\alpha^2 (1-X)$, and so on, $g_k(X) = \alpha^k (1-X)$. Whenever $\alpha > 1$, the weights are increasing in this case.

\subsubsection{Example: Completely Monotone Functions}

The remaining question is when the sufficient \cref{C:cond1,C:cond2} are satisfied? 
For example, one special class of functions where these assumption are satisfied, is the class of functions, where $g_1(X) = -\frac{h(X)}{h'(X)}$ is \emph{completely monotone}, i.e., $(-1)^{k} \dd[k]{g_1(X)}{X} \geq 0$ for all $k \in \N$.\footnote{
	It suffices that $g_1(X)$ is only $T$-times monotone, which is less restrictive, but perhaps harder to verify.
}
This includes many functions, including linear $h(X)$, power function $h(X) = \sqrt[\alpha]{1-X}$, but also many other natural functions. For example, the following functions are all completely monotone: 
$g(X)=\alpha (1-X^m)$, $g(X) =\alpha (1-X)^m$, for all $m \in \N$, $g(X) = \alpha\left( (X+\gamma)^s-(1+\gamma)^s \right)$ for all $s<0, \gamma>0$, $g(X) = \alpha \left[ e^{-rX}-e^{-r} \right]$ for all $r>0$, and $g(X) = -\alpha \log(X)$, all with any $\alpha > 0$. Also, all sums and products of completely monotone functions are completely monotone.

\subsubsection{Example: An Oligopoly with Logarithmic Demand} \label{A:logdemand}

Let me also provide an example where the analysis can be directly extended, even if $g_1(X)$ function is not monotone.
Let us take an oligopoly with $n_1$ leaders and $n_2$ followers, with inverse demand function $P(X)=1-\log X$, and marginal cost $c=1$.\footnote{When the demand function would be linear, this would be exactly the model in \cite{daughety_beneficial_1990}.} Then $h(X) = -\log X$ and $g(X) = -\frac{h(X)}{h'(X)} = -X \log X$, which is not monotone. Therefore, we have to verify the \cref{C:cond1} directly. Applying the recursive rule gives,
\begin{align*}
	f_2(X) &= X, \\
	f_1(X) &= X (1+n_2 \log X), \\
	f_0(X) &= X (1+(n_1+n_2+n_1 n_2) \log X + n_1 n_2 (\log X)^2).
\end{align*}
\Cref{F:ex22log} depicts the functions.
The highest root of $f_1(X)$ is the solution to $1+n_2 \log X=0$, which is
$\uX_1 = e^{-\frac{1}{n_2}} > \uX_2=0$, and $f_1(X)$ is negative in $[0,\uX_1]$ and strictly increasing in $[\uX_1,1]$.
\begin{figure}[!ht]
	\centering
	\includegraphics[trim={0 22pt 0 15pt},clip,width=0.6\linewidth]{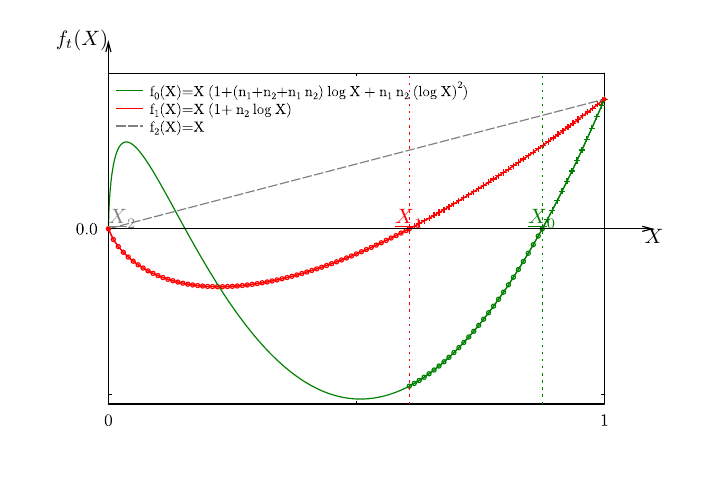}
	\caption{Illustration of \cref{C:cond1} in the case $\bn = (2,2)$ and $g(X)=-X \log X$. Line segments marked with circles are negative and line segments marked with pluses are strictly increasing.
	}
	\label{F:ex22log}
\end{figure}

Now, the expression in the parenthesis of $f_0(X)$ is a quadratic function of $\log X$ that has two roots, both in $(0,1)$. The function $f_0$ is strictly increasing above its highest root, so to verify \cref{C:cond1} it suffices to verify that at $f_0(\uX_1) <0$. Since $\log \uX_1 = -\frac{1}{n_2}$, this is equivalent to
\[
f_0(\uX_1)
= 
\uX_1
\frac{1}{n_2}
\left(
n_2-n_1-n_2-n_1 n_2 + n_1
\right)
=
-n_1
\uX_1<0
\]
The highest root of $f_0(X)$ is the solution to 
\[
1+(n_1+n_2+n_1 n_2) \log X + n_1 n_2 (\log X)^2 = 0,
\]
which is a quadratic equation of $\log X$ and gives 
\[
\log \uX_0
= \frac{\sqrt{(n+p)^2 - 4 p}-(n+p)
}{2 p} 
\;\;\iff\;\;
\uX_0
= e^{\frac{\sqrt{(n+p))^2 - 4 p}-(n+p)
	}{2 p}},
\]
where $n=n_1+n_2$ and $p=n_1 n_2$. Then $\uX_0>\uX_1$, $f_0(X)<0$ between and strictly increasing above $\uX_0$. The equilibrium effort is $X^* = \uX_0$

For example, if $n_1=n_2=2$, then $X^* = e^{\frac{\sqrt{3}}{2}-1} \approx 0.8746$. 
On the other hand, if all $4$ players were to make simultaneous decisions, then the equilibrium would be $X^* = e^{-\frac{1}{4}} \approx 0.7788$. 
If $\bn=(1,3)$ or $\bn=(3,1)$, then $X^* = e^{\frac{\sqrt{37}-7}{6}} \approx 0.8582$.

\subsubsection{Example: Payoffs That Violate \Cref{C:cond1,C:cond2}} \label{A:cond_violated}

In this example, I construct a simple modification of Tullock payoffs, which may not satisfy \cref{C:cond1,C:cond2}.
Suppose that the payoff function is $h(X) = \sqrt[\alpha]{\frac{1}{X} -1}$ for $\alpha>0$. Then
$g_1(X) =\alpha X (1-X)$ and $g_2(X) = \alpha^2 X (2X-1) (1-X)$.

Suppose that we have the simultaneous two-player contest $\bn=(2)$. Then the equilibrium in pure strategies exists only if $\alpha > \frac{1}{2}$. It is useful to compare this example with the Tullock payoffs, where $\alpha=1$. In this case, actions are strategic substitutes if their sum $X$ is above $\frac{1}{2}$ and strategic complements at low levels. The equilibrium with two players was the knife-edge case (which I assumed away by assuming $n>2$), where the actions are neither substitutes nor complements. 
When $\alpha<1$, the region where efforts are strategic complements is expanded so that the equilibrium is in the region of strategic complements, which is an incentive to reduce actions in equilibrium. If the strategic complementary is strong enough, the pure-strategy equilibrium vanishes.
When $\alpha>\frac{1}{2}$, the total equilibrium effort is $X^*=\frac{2\alpha-1}{2\alpha}$ and individual efforts are $x_i^* = \frac{X^*}{2}$.

Next, consider the sequential two-player contest $\hbn = (1,1)$. Then the equilibrium existence requires even higher $\alpha$, as we have an additional ``encouragement effect'' with complements: player 1 whose action is observable wishes to reduce the action of player 2 by choosing lower effort. In particular, the equilibrium exists if and only if $\alpha 
> 2 (\sqrt{2} - 1) \approx 0.82843$. The total equilibrium effort is
\[
X^* 
= \frac{3 \alpha - 2 + \sqrt{\alpha^2 + 4 \alpha - 4}}{4 \alpha}.
\]
Individual effort of player $2$ is $x_2^* = g_1(X^*)$ and therefore for player $1$ it is $x_1^* = X^* - g_1(X^*)$.

For concreteness, let us take $\alpha = 0.9$. 
Then in simultaneous contest, $X^*\approx 0.4444$ and $x_1^*=x_2^*\approx 0.2222$.
In sequential contest, $X^* \approx 0.3723$ and individual efforts $x_1^* \approx 0.1620$ and $x_2^* \approx 0.2103$.
As we see, in sequential, both players reduce their efforts compared to the simultaneous version, and there is a last-mover advantage.

These calculations illustrate how results change when \cref{C:cond1,C:cond2} are not satisfied. For a small $\alpha$, the inverted best-response functions are not well-behaved (\cref{C:cond1} not satisfied), and therefore the equilibrium may even not exits. For $\alpha$ slightly below $1$, efforts in two-player contest are not strategic substitutes, i.e., \cref{C:cond2} is not satisfied. In this case, the Information Theorem and the Earlier-Mover advantage fail to hold.

These calculations also illustrate the importance of implicit assumptions in earlier results, like \cite{daughety_beneficial_1990} beneficial concentration result. \cite{daughety_beneficial_1990} 
showed that in an oligopoly model with identical firms and linear demand, dividing the oligopolists between two periods, instead of having them move in one period, both increases the industry concentration in terms of Herfindahl-Hirschman Index (HHI) and also increases the consumer surplus as it increases total quantity. My results show that \cite{daughety_beneficial_1990} result holds much more generally than was previously known. But, the analysis also shows its limitations. Assuming linear demand function implies that quantities are necessarily higher-order strategic substitutes. If they are not, such as in the example here, the conclusion may be reversed. Indeed, in the numeric example with two players and $\alpha = 0.9$, we got that simultaneous version of the game is less concentrated, as both players get equal market share,  and ensures higher consumer surplus, as the total quantity is larger.

\subsubsection{Example: Direct Substitutes, but Indirect Complements} \label{A:examples}

This final example shows that efforts can be direct substitutes in the standard sense, but complements due to the indirect effects. Suppose that $u_i(\bx)= x_i h(X)$, where for some $a>0$ and all $X < \sqrt{a}$, $h(X)$ is
\begin{equation}
	h(X)
	= e^{ -\sqrt{a} \tan^{-1} \left( \frac{X}{\sqrt{a-X^2}} \right) },
\end{equation}
and for $X \geq \sqrt{a}$, $h(X)=0$.\footnote{The function is not continuous at $\sqrt{a}$, so technically it violates the assumptions in the paper, but this does not affect the conclusions.} The advantage of this function is that we can compute $g_k$ functions for $k \leq 3$ as follows:
\begin{align*}
	g_1(X) &= -\frac{h(X)}{h'(X)} = \sqrt{1-\frac{X^2}{a}}, \\
	g_2(X) &= -g_1'(X) g_1(X)
	= - \frac{-\frac{2X}{a} \frac{1}{2}}{\sqrt{1-\frac{X^2}{a}}} \sqrt{1-\frac{X^2}{a}}
	= \frac{X}{a}, \\
	g_3(X) &= -g_2'(X) g_1(X)
	= -\frac{1}{a} \sqrt{1-\frac{X^2}{a}}.
\end{align*}
Therefore $g_1(X)>0$ and $g_2(X)>0$ for all $X < \sqrt{a}$, so the actions are strategic substitutes in the classical sense. However, $g_3(X) < 0$ for all $X < \sqrt{a}$, which means that they are not higher-order strategic substitutes.
When we are looking for simultaneous game, we get equilibrium $X^* = \frac{\sqrt{a}n}{\sqrt{a+n^2}}$, which is a well-behaving relationship with $X^*$ strictly increasing in $n$ and $\lim_{n \to \infty} X^* = \sqrt{a}$.

In two-period model $\bn=(n_1,n_2)$, pure-strategy equilibrium with all players active may not exist. The issue is that the strategic substitutability effect, quantified by $g_2(X) = \frac{X}{a}$, may become large as $X$ increases. Therefore the benefit from the discouragement effect for the first-movers may be so large that they end up at a corner solution. However, when $a$ is large enough, then we still get interior solutions, and the model behaves as the intuition from the direct discouragement effect would predict. Each disclosure increases total effort, and homogeneity increases total effort. Total effort in contest $(1,3)$ is the same as in $(3,1)$ and both are higher than in less informative contest $(4)$, and both are lower than in more homogeneous contest $(2,2)$.

When we add the third period, the results may change because the indirect discouragement effect is negative, captured by term $g_3(X) <0$. To see this effect in practice, let us compare contest $\bn = (3,3)$ and $\hbn=(4,1,1)$. In both contests, there are six players. There are also the same number, nine, direct observations of other players' efforts. The difference is that in contest $\hbn$, there are also four indirect observations of efforts. Under the higher-order strategic substitutes assumption, this would lead to higher total effort in $\hbn$, and it is easy to check that this is indeed the case with Tullock payoffs. But under the functional form assumptions here, the order is reversed.
For example, if $a=25$, then total equilibrium effort in $\bn$ is $X^* = \frac{75}{17} \approx 4.4118$, whereas total equilibrium effort in $\hbn$ is $\hX^* = \frac{365}{13 \sqrt{41}} \approx 4.3849$.

\subsection{Tullock Contests with Quadratic Costs} \label{A:quadratic}

The linearly multiplicative payoffs exclude the possibility of nonlinear cost of effort and it is natural to ask whether and to what extent the analysis applies in the case of nonlinear costs. Suppose that the payoff function is
\begin{equation}
	u_i(x_i,X)
	= \frac{x_i}{X} - c(x_i), \;\;\; x_i \in \gX_i = \R_+.
\end{equation}
\Cref{E:FOClinmult} in the previous subsection illustrates that the reason why we can easily express individual equilibrium action $x_i$ as a function of the total action $X$ is that in the linearly multiplicative case, the first-order condition is linear in individual action. This property fails for general cost function $c(x_i)$.
However, in a special case with quadratic costs, $c(x_i) = x_i + \frac{\beta}{2} x_i^2$, the first-order condition remains linear and therefore the analysis is relatively tractable. This allows studying convex ($\beta>0$) or concave ($\beta<0$) costs function.\footnote{Note that we can normalize the coefficient in front of $x_i$ to zero without loss of generality for the same reason as in the main text---we can always scale all units by the same factor.}

In this specification, the first-order optimality condition for players in the last period is
\[
\frac{1}{X} - \frac{x_i}{X^2} - 1 - \beta x_i = 0
\iff
x_i = \frac{X(1-X)}{1+\beta X^2}.
\]
Adding up over $i \in \prt_T$ gives us the inverted best-response function
\[
f_{T-1}(X) = X - n_T \frac{X(1-X)}{1+\beta X^2}.
\]

\paragraph{Example: simultaneous contest.} A necessary condition for interior equilibrium in pure strategies is
\[
f_0(X^*) = X^* - n \frac{X^* (1-X^*)}{1+\beta (X^*)^2} = 0
\iff
\beta (X^*)^2 + n X^* + 1 - n = 0.
\]
If $\beta = 0$, then $X^* = \frac{n-1}{n}$. Assuming $\beta \neq 0$, we get two roots, i.e., two candidates for equilibria. Only the higher of the two can be non-negative real number, which is
\[
X^* =
\frac{
	-n
	+ \sqrt{n^2 + 4 \beta (n-1)}
}{2 \beta}.
\]
It is straightforward to check that the solution exists for any $n \geq 2$ if $\beta > -1$. Of course, $\lim_{\beta \to 0} X^* = \frac{n-1}{n}$.

We can continue with similar arguments when $T>1$. The first-order condition for players in period $t<T$ is
\[
\frac{1}{X} - \frac{x_i}{X^2} \frac{1}{f_t'(X)} - 1 - \beta x_i = 0
\iff
x_i  = \frac{X(1-X) f_t'(X)}{1+ \beta X^2 f_t'(X)}.
\]
Combining these conditions for $i \in \prt_t$ gives us
\[
X_{t-1} = X_t - n_t \frac{X(1-X) f_t'(X)}{1+ \beta X^2 f_t'(X)}
\Rightarrow
f_{t-1}(X) = f_t(X) - n_t \frac{X(1-X) f_t'(X)}{1+ \beta X^2 f_t'(X)}.
\]
Therefore we can now define the inverted best-response functions recursively with this rule and get a necessary condition for equilibrium as follows
\begin{align*}
	f_T(X) &= X,  \\
	f_{t-1}(X) &= f_t(X) - n_t \frac{X(1-X) f_t'(X)}{1+ \beta X^2 f_t'(X)}, \\
	f_0(X) &= 0.
\end{align*}
These expressions are more complicated and do not simplify to the characterization with measures of information, described in \cref{S:information}, but it is easy to see that for $\beta$ close enough to zero, the characterization result still applies as do most of the other results in the paper. The following proposition summarizes these claims.

\begin{proposition}
	For each $\bn$, there exists $\underline{\beta}< 0 < \overline{\beta}$ such that the contest $\bn$ with quadratic cost function with any parameter $\beta \in (\underline{\beta},\overline{\beta})$ has a unique equilibrium, characterized in the same way as in \cref{T:characterization}.
	Moreover, \cref{P:earliermover} (earlier-mover advantage) and all conclusions in \cref{C:information} apply, except (2) (independence of permutations).
\end{proposition}

The proof is a straightforward continuity argument. In each step of the proof of \cref{T:characterization}, \cref{C:information}, and \cref{P:earliermover}, the inequalities were strict. As the $f_t$ functions are continuous in $\beta$, the same inequalities continue to hold at least in some interval around $\beta=0$. The only exception is the independence of permutations in \cref{C:information}, where the result relied on the fact that $f_0$ is exactly equal when swapping two groups, and as the following example illustrates, this result does not hold anymore.

\paragraph{Example: two-period contest.} When $T=2$, we get
\begin{align*}
	f_1(X) &= X - n_2 \frac{X(1-X)}{1+ \beta X^2}, \\
	f_0(X) &
	= f_1(X) - n_1 \frac{X(1-X) f_1'(X)}{1+ \beta X^2 f_1'(X)}
	= 0.
\end{align*}
The highest root $X^*$ of the last equation is still the equilibrium effort in the contest and the value is straightforward to compute numerically. \Cref{F:fpermqcosts} depicts the function $f_0(X)$ with seven-player contests $(1,6)$ and $(6,1)$ under linear, convex, and concave costs. Remember that the highest root of $f_0(X)=0$ is the total equilibrium effort. We can make a few observations from this figure. First, the figure confirms the conclusion from above: as long as $\beta$ is close enough to zero, the qualitative properties of $f_0$ are unchanged. Second, if $\beta>0$, the equilibrium effort is typically lower than in the linear case, which is natural as the costs are now higher. Similarly, if $\beta<0$, the equilibrium effort is higher than in the linear case.

Third, independence of permutations fails whenever $\beta \neq 0$ (the dashed lines do not coincide and the dotted lines do not coincide). This is intuitive. In the linear cost case, the total effort was the same first-mover $(1,6)$ and last-mover contest $(6,1)$, but the effort distribution was not. With a single leader, the leader exerts much higher effort than the followers. In the single follower case, all leaders exert the same, a quite low effort. If costs are quadratic, these two have different costs and therefore would not be equivalent anymore.
\begin{figure}
	\centering
	\begin{subfigure}[b]{0.48\textwidth}
		\centering
		\includegraphics[width=\textwidth]{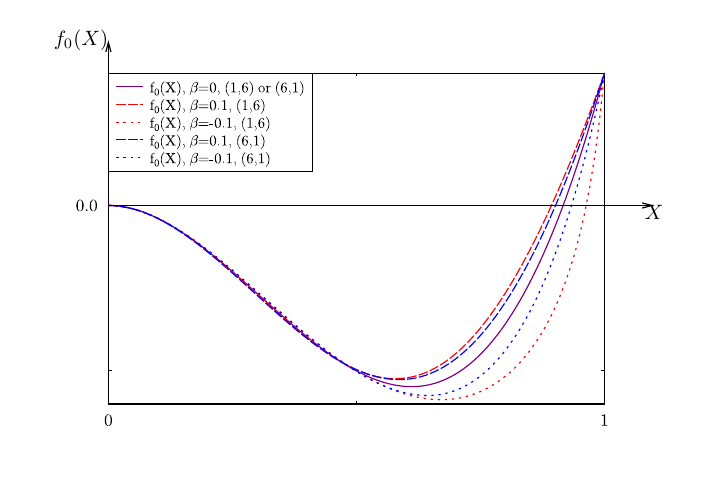}
		\caption{Function $f_0$ for contests $(1,6)$ and $(6,1)$ with linear ($\beta=0$), convex ($\beta=0.1$), and concave ($\beta=-0.1$) costs.}
		\label{F:fpermqcosts}
	\end{subfigure}
	\hfill
	\begin{subfigure}[b]{0.48\textwidth}
		\centering
		\includegraphics[width=\textwidth]{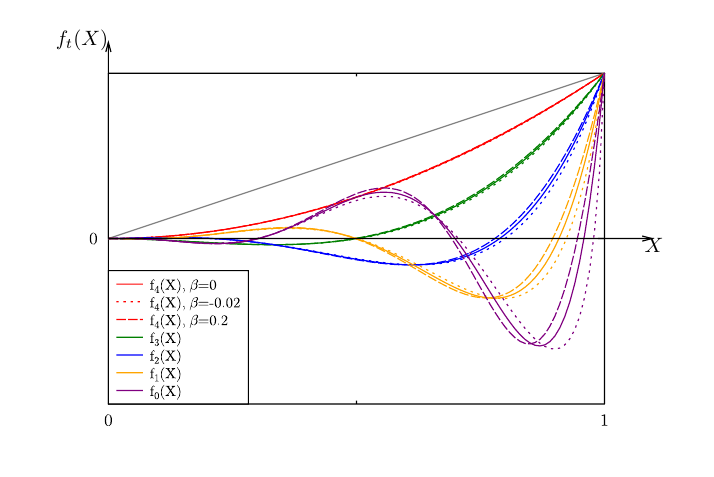}
		\caption{$f_t$ functions for five-player sequential contest $\bn=(1,1,1,1,1)$ with linear ($\beta=0$), convex ($\beta=0.2$), and concave ($\beta=-0.02$) costs.}
		\label{F:fqcosts}
	\end{subfigure}
	\caption{Equilibrium characterization with quadratic costs. The highest root of $f_0(X)=0$ is still the total equilibrium effort.}
	\label{F:fqcost_plots}
\end{figure}

Finally, \cref{F:fqcosts} illustrates the case of five-player sequential contest $\bn = (1,1,1,1,1)$. The solid lines are $f_t$ are the $f_t$ functions in the case of linear cost function discussed in the main text ($\beta=0$), the dashed lines correspond to convex cost ($\beta=0.2>0$), and the dotted lines correspond to concave costs ($\beta=-0.02$). The figure illustrates that when $\beta$ is small enough, the ordering of the roots of $f_t$ functions remains unchanged, so that the arguments in the proof of \cref{T:characterization} still hold.

\subsection{Tullock Contests with Heterogeneous Players} \label{A:asymmetric}

We can also extend the analysis to heterogeneous players. For example, players can differ by their valuations for the prize, $v_i$, and costs of efforts $c_i$, so that the payoff functions are the following.\footnote{More generally, in the form $u_i(x_i,X)=x_i h_i(X)$ for some well-behaving functions $h_i$.}
\[
u_i(x_i,X) = \frac{x_i}{X} v_i - c_i x_i.
\]
In this case, the first-order condition of a player in period $T$ is
\[
\dd{u_i(x_i,X)}{x_i}
= \frac{1}{X} v_i - \frac{x_i}{X^2} v_i - c_i
\leq 0,
\]
with equality whenever the optimal $x_i > 0$. We can express $x_i$ as a function of $X$ as
\[
x_i = \max\left\{0, X \left( 1-\frac{c_i}{v_i} X \right) \right\}.
\]
Adding up these expressions gives us $f_{T-1}(X) = X - G_T(X)$, where
\[
G_T(X) = \sum_{i \in \prt_T} \max\left\{0, X \left( 1-\frac{c_i}{v_i} X \right) \right\}.
\]
We see that overall the approach extends with heterogeneous players as well, but there is an added complication. As long as $\frac{c_i}{v_i}$ is different for different players, they choose to be inactive at different levels of total effort $X$. Therefore solving for the equilibrium behavior requires considering the corner solutions.

This issue is even more complicated when considering the behavior of players in earlier periods because the inverted best-response function $f_t$ is not typically continuously differentiable. At points, where some followers choose to become inactive, its slope changes discretely. Therefore players may choose to adjust their efforts just enough so that some followers will stay inactive.

Consider period $t$ and suppose that $X_t$ implies total effort $f_t^{-1}(X_t)$. If player $i$'s optimality condition is such that it implies continuously differentiable $f_t$, it must satisfy the first-order optimality condition
\[
\frac{1}{X} v_i - \frac{x_i}{X^2} v_i \frac{1}{f_t'(X)}  - c_i = 0
\;\;\;
\iff
\;\;\;
x_i = X \left( 1 - \frac{c_i}{v_i} X \right) f_t'(X).
\]
Again, adding up these expressions gives us $f_{t-1}(X)=f_t(X) - G_t(X) f_t'(X)$, where
\[
G_t(X) = \max\left\{
0, X \left( 1 - \frac{c_i}{v_i} X \right)
\right\}.
\]
This suggests a method of finding equilibria in contests with heterogeneous players: either the total equilibrium effort is such that $X^* = \frac{c_i}{v_i}$ for some $i$, so that one player is exactly remaining inactive, or it is $X^*$ such that some players stay active (at least two) and some stay inactive (possibly none), and the equilibrium is determined as a root of $f_0(X^*)=0$, where $f_0,\dots,f_T$ are defined recursively as shown above.

\paragraph{Example: two players.} Suppose that $n=2$. Without loss of generality, let us call one player be strong (type $s$) and normalize $\frac{c_s}{v_s}=1$. The other player is weak (type $w$) with $\frac{c_w}{v_w}=c \geq 1$. In the symmetric case when $c=1$, both simultaneous and sequential game lead to the same equilibrium, where $x_w^*=x_s^*=\frac{1}{4}$ and therefore $X^*=\frac{1}{2}$. However, when $c>1$, we have three possible orderings and straightforward calculations lead to the following conclusions.

In simultaneous contest, both players choose to be active and the total equilibrium effort is $X_{sim}^* = \frac{1}{1+c}$. In sequential contest where weak player moves first (order $w \to s$) both players still stay active, but the total equilibrium effort is $X_{w \to s}^* = \frac{1}{2c} < X_{sim^*}$. Finally, in sequential contest where the strong player moves first (order $s \to w$), the qualitative properties of equilibrium behavior depend on $c$. If $c<2$, i.e., heterogeneity is relatively mild, then both players stay active and the total equilibrium effort is $X^*_{s \to w} = \frac{1}{2} > X^*_{sim}$. But if $c \geq 2$, the strong player deters entry by the weak player by choosing $x_s^* = \frac{1}{c}$. Therefore in this case $x_w^*=0$ and $X^*_{s \to w} = x_s^* = \frac{1}{c} > X^*_{sim}$.\footnote{
	\cite{morgan_sequential_2003} and \cite{serena_sequential_2017} analyze the two-player asymmetric sequential contests (with endogenous order of moves). In this case, the equilibrium can be characterized with the standard backward induction. \cite{xu_three-player_2020} use the approach introduced here to study the three-player asymmetric sequential contests.}

\end{document}